\documentclass[12pt,draftcls,onecolumn]{IEEEtran}

\usepackage{amsmath}
\usepackage{amssymb}
\usepackage{graphicx}
\usepackage{latexsym}
\usepackage{comment,url}
\usepackage{subfigure}
\usepackage{stfloats}
\usepackage{algorithm,algorithmic}

\newcommand{\R}{{\mathcal R}}
\newcommand{\Y}{{\mathcal Y}}
\newcommand{\C}{{\mathcal C}}
\newcommand{\ve}{\varepsilon}
\newcommand{\V}{{\mathcal V}}
\newcommand{\A}{{\mathcal A}}

\newcommand{\W}{{\mathcal W}}
\newcommand{\Q}{{\mathcal Q}}
\newcommand{\X}{{\mathcal X}}
\newcommand{\BB}{{\mathcal B}}
\newcommand{\E}{{\mathcal E}}
\newcommand{\T}{{\mathcal T}}

\renewcommand{\H}{{\mathcal H}}

\newcommand{\diag}{{\rm diag}}

\newcommand{\N}{{\mathcal{N}}}

\newtheorem{definition}{Definition}

\newtheorem{lemma}{Lemma}

\newtheorem{theorem}{Theorem}
\newtheorem{assumption}{Assumption}

\begin{document}
\title{Cooperative Environmental Monitoring \\
for PTZ Visual Sensor Networks:\\ A Payoff-based Learning Approach}

\author{Takeshi Hatanaka,~\IEEEmembership{Member,~IEEE}, Yasuaki Wasa and\\
Masayuki Fujita,~\IEEEmembership{Member,~IEEE} 
\thanks{Takeshi Hatanaka(corresponding author), Yasuaki Wasa
and Masayuki Fujita
are with the Department of Mechanical and Control Engineering,
Tokyo Institute of Technology, Tokyo 152-8550,
JAPAN, {\sf hatanaka@ctrl.titech.ac.jp}}
}


\maketitle

\begin{abstract}
This paper investigates cooperative environmental monitoring
for Pan-Tilt-Zoom (PTZ) visual sensor networks.
We first present a novel formulation of the optimal
environmental monitoring problem, whose objective function is
intertwined with the uncertain state of the environment.
In addition, due to the large volume of vision data, 
it is desired for each sensor 
to execute processing through local computation and communication.
To address the issues, we present a distributed solution to the problem
based on game theoretic cooperative control and payoff-based learning.
At the first stage, a utility function is designed so that the resulting
game constitutes a potential game with potential function
equal to the group objective function,
where the designed utility is shown to be computable 
through local image processing and communication. 
Then, we present a payoff-based learning algorithm
so that the sensors are led to 
the global objective function
maximizers without using any prior information on 
the environmental state.
Finally, we run experiments 
to demonstrate the effectiveness of the present approach.
\end{abstract}

\IEEEpeerreviewmaketitle

\section{Introduction}
\label{sec:1}

Large-scale environmental monitoring 
to reveal environmental states has become crucial
due to recent serious natural disasters including earthquakes, 
tsunamis, nuclear meltdowns, landslides, typhoons/hurricanes and so on. 
In the task, it is in general required 
to collect dense data in real time over widespread environment.
As a solution to the issue, sensor networks have emerged 
over the past few decades and been extensively studied.
Moreover, mobile/robotic sensor networks have also been deeply
investigated as a key technology to enhance data collection efficiency
\cite{DM_RASM12}.

Among a variety of sensors available for 
the monitoring task \cite{DM_RASM12},
this paper focuses on vision sensors.
In particular, we consider a sensor network consisting of
spatially distributed cameras, 
which is  called {\it camera/visual sensor network}
\cite{CS_BK}.
Then, we need to take account of the following
nature of vision sensors:
(i) volume of data tends to be larger than the other sensors, 
(ii) vision sensors do not provide explicit physical data, and
(iii) vision sensors are inherently heterogeneous,
which means that, even if quality of two sensors are the same, 
quality of their measurements on a common point can differ 
in the location of the point relative to the camera frames. 

In this paper, we investigate a distributed/cooperative 
optimal monitoring strategy for a network of Pan-Tilt-Zoom (PTZ) 
cameras by controlling the camera parameters,
which are called {\it actions} in this paper.
Note that 
what is the optimal action in the problem is affected by
the unknown environmental state.
Accordingly, we have to solve the optimization problem under the restriction: (iv) 
each vision sensor has no access to the reward brought about by an action 
before the action is actually executed.
Optimization under (iv) has been deeply
studied in the field of reinforcement learning 
and simulated annealing. 
However, these algorithms are centralized and might not be 
available due to the nature (i).

This paper first formulates a novel
optimal environmental monitoring problem for PTZ visual sensor networks
reflecting the nature (ii) and (iii), 
where we let the objective function rely on the amount of information
contained in the sensed data and the quality of the measurement.
Then, we next present a distributed solution to the problem
leading the sensors to the globally optimal actions
under the restriction of (iv).
To meet the requirements, this paper employs
techniques in {\it game theoretic cooperative control} originally presented in \cite{MAS_SMC09}
since it provides a systematic design procedure of cooperative control
for heterogeneous networks as stated in (iii). 
Following the procedure of \cite{MAS_SMC09}, we first constitute a potential game \cite{MAS_SMC09} with the potential function
equal to the global objective function through an
appropriate utility design technique \cite{GMW_PER10}.
Then, as a technical tool to address (iv), we employ an 
action selection rule called {\it payoff-based learning} 
\cite{MS08_GEB12}--\cite{GHF_ACC12},
where each player chooses his action
based only on the past experienced payoffs.
In particular, we present a novel payoff-based learning algorithm
which guarantees convergence in probability to
the potential function maximizers, which are equal to
the global objective function maximizers.
Finally, we run experiments on a visual sensor network testbed
to demonstrate the effectiveness of the present approach.


\subsection*{Related Works and Contributions}



Due to the nature of vision sensors (i), distributed processing 
over visual sensor networks has been actively
studied in some recent papers.
Cooperative estimation over visual sensor networks
is studied in \cite{SDKFC_SP11}, \cite{HF_TAC13}, \cite{TV_SP11},
distributed localization/calibration is investigated
in \cite{TV_SP11}, \cite{BLCFZ_ACC12}, and distributed sensing strategies
are presented in \cite{DSMFC_SP12}, \cite{ZM_SIAM13}, \cite{SPF_ACC12}.
In particular, the scenarios and approaches in \cite{DSMFC_SP12}, \cite{ZM_SIAM13}
are closely related to this paper and hence
they will be mentioned later.

The objective of this paper is related to 
coverage control \cite{CL_EJC05}--\cite{SJS_TAC12}
whose objective is to deploy mobile sensors efficiently 
via distributed decision-making.
A gradient decent approach widely used 
in the literature \cite{CL_EJC05,BCM_BK09}
is implementable even under the restriction (iv).
However, the approach is not always directly applicable
to the problem of this paper
due to the nature of vision sensors (ii) and (iii).
More importantly, the gradient decent approach leads sensors to
a configuration achieving local maxima of some group objective function,
but such a configuration does not always globally maximize the objective function.

Persistent monitoring is also recently studied e.g. in 
\cite{HSS_CDC08}--\cite{CLD_CDC12},
which differs from coverage in
the perpetual need to cover a changing environment \cite{SSR_ICRA11,CLD_CDC12}.
However, to the best of our knowledge, there are few works
fully taking account of the nature of vision sensors.
In addition, while most of the works \cite{HHGHFS_IFAC08}--\cite{CLD_CDC12}
assume information accumulation/decay models
and availability of the model, this paper does not presume such models.

The papers \cite{MAS_SMC09,DSMFC_SP12,ZM_SIAM13,SJS_TAC12} are most directly 
related to this paper, where the authors investigate
potential game theoretic approaches to coverage control
or collaborative sensing.
The algorithm presented in \cite{MAS_SMC09} guarantees that
players eventually take the globally optimal action with high probability.
However, it presumes availability of 
future payoffs prior to action executions
and hence cannot be implemented under (iv).
\cite{ZM_SIAM13} presents a payoff-based learning algorithm
and applied it to coverage for visual (mobile) sensor networks.
However, the algorithm does not always lead sensors to
the globally optimal actions and the nature (ii) and (iii)
are not fully addressed.
Similar statements are also true for \cite{SJS_TAC12}.
Meanwhile, \cite{DSMFC_SP12} mentions how to use
the visual measurement in the process explicitly.
Though the authors utilize a learning algorithm
assuming a future payoff, they successfully avoid the issue (iv)
by constructing the future 
virtual utility from the estimate of the target states produced by a distributed filter.
However, the approach may limit applications
since there might be no explicit target
in some scenarios.

We finally mention the contribution of the present learning algorithm.
The algorithm is regarded as a variation of
\cite{MS08_GEB12} and \cite{ZM_SIAM13}.
\cite{MS08_GEB12} guarantees that potential function maximizers 
are eventually selected with high probability.
Meanwhile, \cite{ZM_SIAM13} has advantages over \cite{MS08_GEB12} 
that the action selection rule is simpler and convergence in probability
is rigorously guaranteed, but it does not always lead sensors
to potential function maximizers.
The contribution of the present algorithm is to embody advantages of these two algorithms,
i.e. guarantees convergence in probability to
potential function maximizers  while maintaining the
simple structure of \cite{ZM_SIAM13}.

The contributions of this paper are summarized as follows:
\begin{itemize}
\item a novel problem formulation of environmental monitoring for PTZ visual sensor networks
taking account of the nature of vision sensors (ii) and (iii) is presented,
\item a novel simple payoff-based learning algorithm for potential games 
guaranteeing convergence in probability to
the potential function maximizers is proposed, and
\item the approach is demonstrated through
experiments, while such efforts are not always fully made
in the existing works on game theoretic cooperative control.
\end{itemize}

\section{Visual Sensor Networks and Environment}
\label{sec:2}

\subsection{Visual Sensor Networks and Environment}
\label{sec:2.1}

\begin{figure}
\begin{center}
\begin{minipage}[t]{8cm}
\includegraphics[width=8cm]{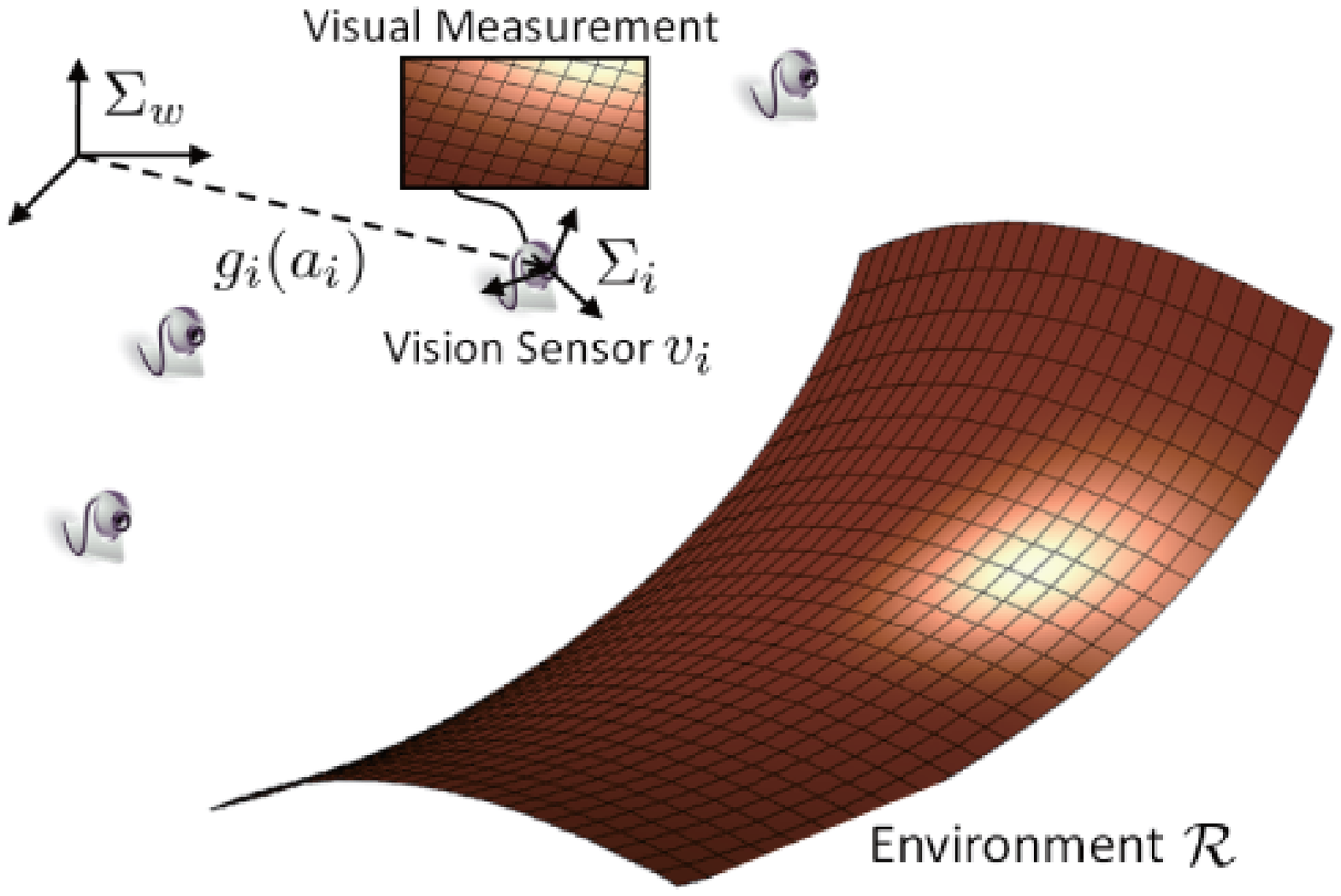}
\caption{Targeted scenario}
\label{fig:1}
\end{minipage}
\hspace{1.5cm}
\begin{minipage}[t]{5cm}
\begin{center}
\includegraphics[width=5cm]{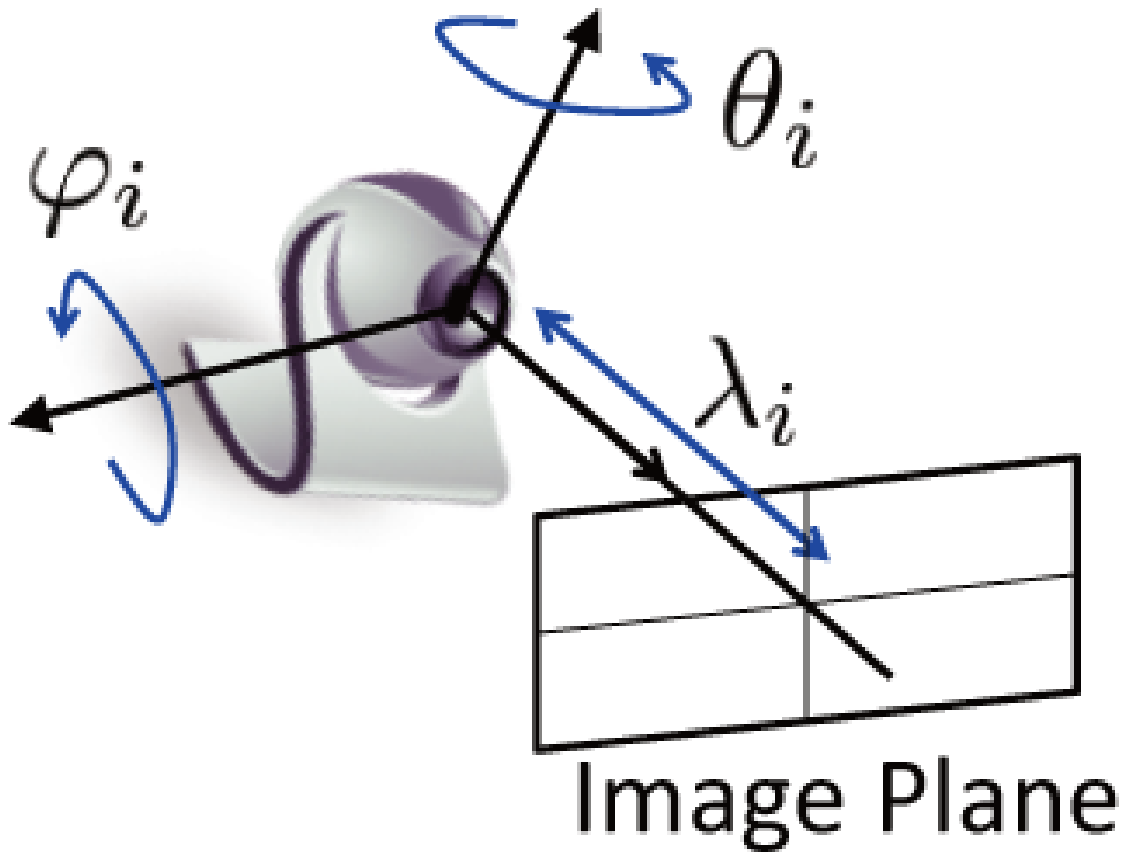}
\caption{Actions of vision sensor $v_i$}
\label{fig:2}
\end{center}
\end{minipage}
\end{center}
\end{figure}

In this paper, we consider the situation illustrated in Fig. \ref{fig:1}, where
$n$ Pan-Tilt-Zoom (PTZ) vision sensors $\V = \{v_1, \cdots, v_n\}$ 
monitor environment
modeled by a collection of $m$ polygons $\R = \{r_1, \cdots, r_m\}$. 
Let the set of position vectors of all points in $r_j\in \R$
relative to a world frame $\Sigma_w$
be denoted by $\Q_j$.
In the following, we also use the notation $\Q = \cup_{r_j \in \R}\Q_j$.

Suppose that each PTZ vision sensor $v_i \in \V$ can adjust
its horizontal (pan) angle $\theta_i \in \varTheta_i \subseteq [-\pi, \pi]$, 
vertical (tilt) angle $\varphi_i \in \varPhi_i \subseteq [0, \pi]$
and focal length $\lambda_i \in \varLambda_i$ (Fig. \ref{fig:2}), where
$\varTheta_i, \varPhi_i$ and $\varLambda_i$ are assumed to be finite sets.
Throughout this paper, the notation 
\begin{equation*}
a_i = (\theta_i, \varphi_i, \lambda_i)
\in \A_i := \varTheta_i \times \varPhi_i\times \varLambda_i
\end{equation*} 
is called an {\it action} of sensor $v_i \in \V$,
and 
$a = (a_i)_{v_i \in \V} \in \A := \A_1 \times \cdots \times \A_n$
is called a {\it joint action}.
A collection of actions other than 
$v_i$ is denoted as $a_{-i}:=(a_{1},\cdots,a_{i-1},a_{i+1},\cdots,a_{n})$.

Once an action $a_i$ is fixed,
the orientation of sensor $v_i$'s frame $\Sigma_i$
relative to $\Sigma_w$
and its maximal view angle
are uniquely determined, which are respectively denoted by 
$R_i(a_i)\in SO(3) := \{R\in {\mathbb R}^{3\times 3}|\ 
R^TR = I_3,\ \det(R) = +1\}$ and $\bar{\beta}_i(a_i)$.
The position of the origin of $\Sigma_i$ relative to $\Sigma_w$
is also denoted by $p_i \in {\mathbb R}^3$.
Then, the pose of sensor $v_i$
is represented as $g_i(a_i) = (p_i, R_i(a_i))\in SE(3) := {\mathbb R}^3 \times SO(3)$.
In this paper, we assume that each $v_i\in \V$ is already calibrated
and has knowledge on the pose $g_i(a_i)$
for all $a_i\in \A$, and
$\Q_j = \{q\in {\mathbb R}^3|\ A_{{\rm eq},j}q = 1,\ A_{{\rm ieq},j}q \leq {\bf 1}\}$ 
for all $r_j \in \R$, where ${\bf 1}$ is a vector 
whose elements are all equal to $1$.

We also assume that, when each sensor $v_i\in \V$ takes action $a_i\in \A_i$,
the actions selectable at the next round are constrained by
a subset $\C_i(a_i) \subseteq \A_i$ satisfying
the following assumptions which are in general satisfied in the scenario of this paper.
\begin{assumption}{\rm 
\label{ass:1}
The function $\C_i: \A_i \rightarrow 2^{\A_i}$  satisfies:
\begin{itemize}
\item 
For any $v_i \in \V$, $a_i\in \A_i$ and $a'_i\in \A_i$, 
the inclusion $a'_i\in \C_i(a_i)$ holds iff
$a_i\in \C_i(a'_i)$.
\item 
For any $v_i \in \V$ and any actions $a_i, a'_i \in {\mathcal 
	   A}_i$, there exists a sequence of actions
$a_i = a_i^1, a_i^2, \cdots, a^{n_f}_i = a_i'$
satisfying $a_i^{\iota} \in \C_i(a_i^{\iota-1})$ for all $\iota \in \{2,\cdots, 
 n_f\}$.
\item For any $v_i\in \V$ and $a_i \in \A_i$, 
the number of elements in $\C_i(a_i)$ 
is greater than or equal to $3$.
\end{itemize}
}
\end{assumption}

\subsection{Visual Measurements and Communication Structure}

\begin{figure}
\begin{center}
\begin{minipage}[t]{6cm}
\begin{center}
\includegraphics[width=5cm]{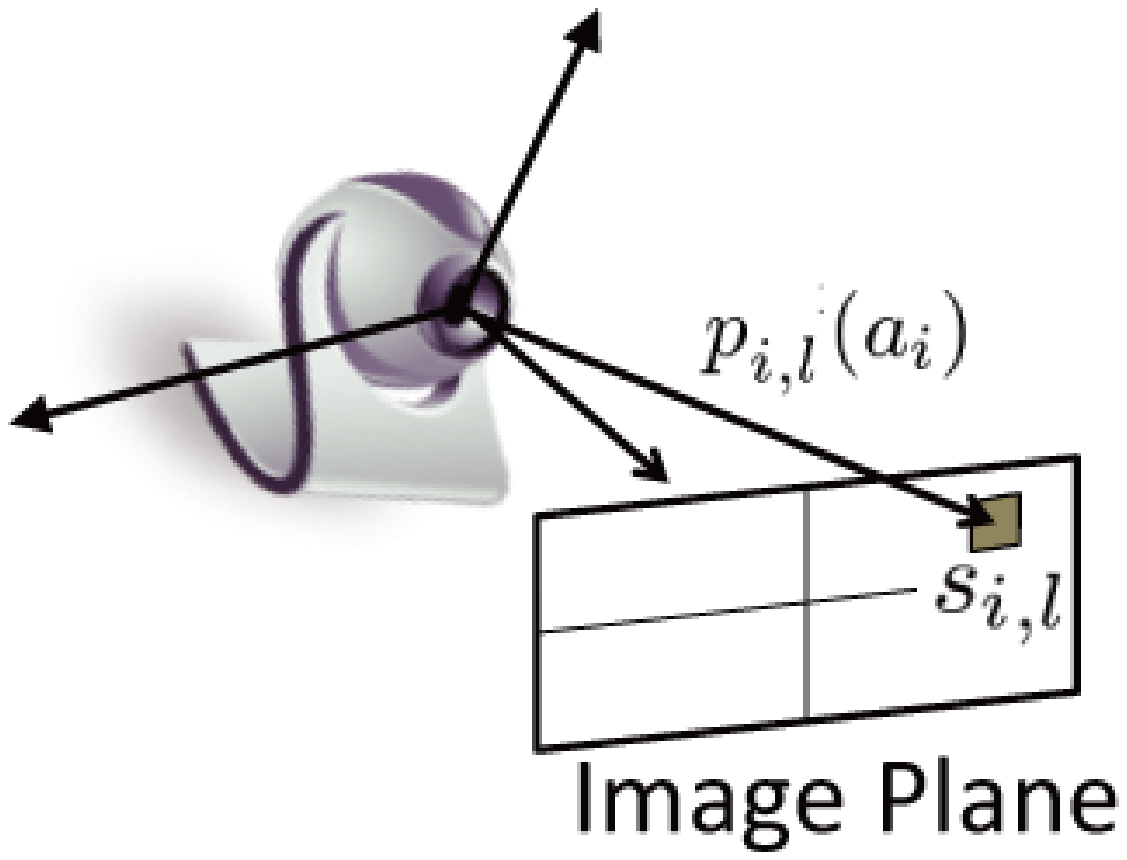}
\caption{Position vector $p_{i,l}(a_i)$ of pixel $s_{i,l}$}
\label{fig:3}
\end{center}
\end{minipage}
\hspace{2cm}
\begin{minipage}[t]{7.5cm}
\begin{center}
\includegraphics[width=7.5cm]{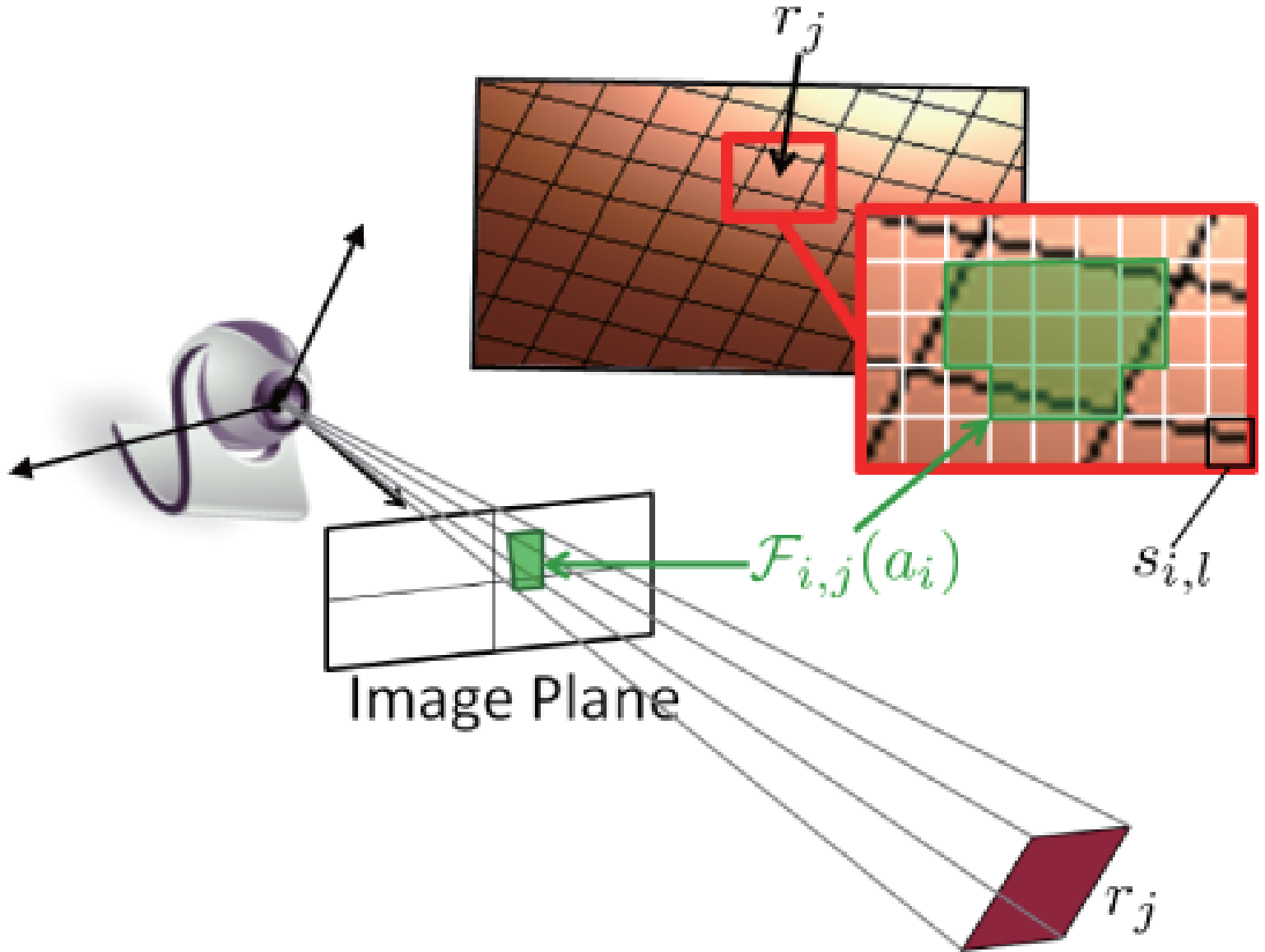}
\caption{Image of set ${\mathcal F}_{i,j}(a_i)$}
\label{fig:4}
\end{center}
\end{minipage}
\end{center}
\end{figure}

This subsection defines visual measurements
of each vision sensor $v_i\in \V$ and communication structures among sensors in $\V$.

Let us first denote the pixels of vision sensor $v_i \in \V$
by ${\mathcal S}_i := \{s_{i, l}|\ l \in \{1, \cdots, S_i\}\}$,
and the position vector of the center of
$s_{i,l}\in {\mathcal S}_i$ relative to $\Sigma_i$ by 
$p_{i,l}(a_i)$ (Fig. \ref{fig:3}).
Then, once an action $a_i$ is fixed, $v_i$
obtains visual measurements (raw data) $y_{i,l}$ for each $s_{i, l} \in {\mathcal S}_i$.
Here, $y_{i,l}$ is a 3D vector for an RGB color image whose elements take integers in
$\{0, \cdots, 255\}$,
and $y_{i,l} \in \{0, \cdots, 255\}$ for a grey-scale image.
Note that each $y_{i,l}$ is provided by either of polygons $r_j \in \R$.

%
In this paper, a point $q \in \Q$ is said to be visible
from sensor $v_i \in \V$ with action $a_i\in \A$ if 
\begin{eqnarray}
&&{\rm atan}(\|[b_x\ b_y]\|/b_z)
\leq \bar{\beta}_i(a_i),\ 
[b_x\ b_y\ b_z]^T := R_i^T(a_i)(q-p_i)
\nonumber
\end{eqnarray}
and there exists no pair of 
$q' \in \Q$ and $\alpha \in [0, 1)$ such that
$\alpha (q - p_i) = q' - p_i$.
%
By using the notion, we also define the set of pixels
of $v_i \in \V$ with $a_i \in \A_i$
capturing $r_j \in \R$,
which is denoted by ${\mathcal F}_{i,j}(a_i)$ (Fig. \ref{fig:4}).
Here, a pixel $s_{i,l} \in {\mathcal S}_i$ is a member of ${\mathcal F}_{i,j}(a_i)$
if and only if there exists a visible point $q \in \Q_j$ such that
the point $q$ projected onto the image plane is equal to 
the center of the pixel $s_{i,l}$.
Due to the knowledge of $g_i(a_i)$ and $\Q_j$, each $v_i\in \V$
can obtain ${\mathcal F}_{i,j}(a_i)$ for every $a_i \in \A_i$.
To be precise, a pixel $s_{i,l}$ is included in ${\mathcal F}_{i,j}(a_i)$
if $r_j$ satisfies
\begin{equation}
\alpha_{i,j}(a_i) = \frac{1-A_{{\rm eq},j}p_i}{A_{{\rm eq},j}R_i(a_i)p_{i,l}(a_i)} \in 
(0, \infty),\ 
A_{{\rm ine},j}\left(\alpha_{i,j}(a_i)R_i(a_i)p_{i,l}(a_i) + p_i\right)\leq {\bf 1}
\label{eqn:visibility}
\end{equation}
and the parameter $\alpha_{i,j}(a_i)$ is minimal
among all polygons satisfying (\ref{eqn:visibility}).

In addition, a polygon $r_j \in \R$ is said to be a {\it visible polygon} 
from sensor $v_i \in \V$ with $a_i\in \A_i$ if $|{\mathcal F}_{i,j}(a_i)|\geq 1$,
where $|{\mathcal F}|$ specifies the number of elements
of a finite set ${\mathcal F}$.
We also denote the set of all visible polygons from $v_i$ with 
$a_i$ by $\R_i(a_i) \subseteq \R$.
In addition, when a joint action $a$ is selected,
the set of sensors capturing $r_j$ as a visible polygon
is denoted by $\V_j(a) \subseteq \V$.
%

%

We also model the communication structure among sensors 
by an undirected graph $G = (\V,\ \E)$ with $\E \subseteq \V 
\times \V$.
The set of all sensors whose information is available for 
 $v_i\in \V$ is also denoted as 
$\N_i := \{v_j\in \V|\ (v_i,v_j) \in \E\}$.
In this paper, we use the following assumption.
\begin{assumption}
\label{ass:2}
A pair $(v_{i},v_{i'}),\ i\neq i'$ satisfies $(v_{i},v_{i'}) \in \E$
if there exist $a_{i}\in \A_{i}$, $a_{i'} \in \A_{i'}$ and 
$r_j \in \R$ such that $r_j \in \R_{i}(a_{i}) \cap \R_{i'}(a_{i'})$.
\end{assumption}
This assumption means that if any pair of two sensors 
can capture a common polygon then they need to communicate with each other,
which is essentially similar to \cite{SDKFC_SP11} and the only slight difference in
description stems from whether multi-hop communication is taken into account.

\section{Global Objective and Utility Function}
\label{sec:3}

\subsection{Global Objective Function}
\label{sec:3.1}

Let us formulate the global objective function
$W(a)\in [0,\infty)$ to be maximized by vision sensors $\V$.
For this purpose, we first introduce a function $W_{i,j}(a_i)\in [0,\infty)$
evaluating the value of measurements $\Y_{i,j}(a_i)$  of sensor $v_i \in \V$ about polygon $r_j \in \R$.

Let us assume that the function $W_{i,j}(a_i)$ relies on
(a) how much information $r_j\in \R$ contains, and
(b) quality of the image.
Formally, if the quantitative values of factors (a) and (b) are denoted by
$I^{\rm info}_{i,j}(a_i)\in [0,\infty)$ and 
$I^{\rm qual}_{i,j}(a_i)\in [0,\infty)$ respectively,
the function $W_{i,j}(a_i)$ is described as
\begin{equation}
W_{i,j}(a_i) = \tilde{W}_{i,j}(I^{\rm info}_{i,j}(a_i), I_{i,j}^{\rm qual}(a_i)).
\label{eqn:reward}
\end{equation}
In general, the function $\tilde{W}_{i,j}(\cdot,\cdot)$ is
non-decreasing with respect to
both $I^{\rm info}_{i,j}(a_i)$ and $I^{\rm qual}_{i,j}(a_i)$,
and the equation
$I^{\rm info}_{i,j}(a_i) = I^{\rm qual}_{i,j}(a_i) =0$ holds
if $r_j$ is not visible from $v_i$, i.e. $r_j \notin \R_i(a_i)$.

The functions $I^{\rm info}_{i,j}(a_i)$ and $I^{\rm qual}_{i,j}(a_i)\in [0,\infty)$
respectively play roles similar to the density function and the sensing performance function 
in coverage control \cite{CL_EJC05}, \cite{BCM_BK09}. 
However, there are some differences.
Since vision sensors do not provide apparent physical quantity
like temperature or pressure, 
we need to extract the amount of information
contained in the raw data $\Y_{i,j}(a_i)$, which makes the selection 
of $I^{\rm info}_{i,j}(a_i)$ non-trivial.
However, fortunately, there are rich literature on 
information extraction from visual measurements,
and we can freely choose one of them
depending on the targeted scenario.
Some examples will be shown in the next subsection.
However, such quantities can be extracted after gaining the visual measurement
and, moreover, the function $I^{\rm info}_{i,j}(a_i)$ is dependent on
the state of the highly uncertain environment regardless of its selection.
Hence, we cannot assume availability of 
the value $I^{\rm info}_{i,j}(a_i)$ prior to execution of $a_i$.
Due to the problem, we will present a solution
using only the past experienced values of $I^{\rm info}_{i,j}(a_i)$
and $I^{\rm qual}_{i,j}(a_i)$ in the subsequent sections.


For visual sensor networks, the quality $I^{\rm qual}_{i,j}(a_i)$ is determined 
not only by the distance between $v_i$ and $r_j$
but also by their relative pose and the focal length $\lambda_i$,
which makes the function complex.
However, since the present solution does not
require to model the function differently from 
coverage control \cite{CL_EJC05}, \cite{BCM_BK09},
it is sufficient to evaluate the quality 
after gaining the image.
For example, using the fraction over the image that $r_j$ occupies as
$I^{\rm qual}_{i,j}(a_i) = f^{\rm qual}(|{\mathcal F}_{i,j}(a_i)|/S_i)$
with an increasing function $f^{\rm qual}$
satisfying $f^{\rm qual}(0) = 0$ can be a useful option.
%

%

We next consider the reward $W_j(a)\in [0,\infty)$ 
provided from environment $r_j$
to not a single sensor $v_i$ but the visual sensor network $\V$.
%
%
We assume that $W_j(a)$ is a function of
$W_{i,j}(a_i)$ only for vision sensors in $\V_j(a)$ capturing $r_j$ as
%
\begin{equation}
W_{j}(a) = \tilde{W}_j((W_{i,j}(a_i))_{v_i \in \V_j(a)})
\label{eqn:total_reward}
\end{equation}
with $W_j(a) = 0$ if $\V_j(a) = \emptyset$
and that vision sensors $\V$
share the information of the function $\tilde{W}_j$.
In the following, we show only two typical selections
of such functions.
The first option is 
\begin{equation}
 W_j(a) = \max_{v_i \in \V_j(a)}W_{i,j}(a_i)
\label{eqn:total_reward1}
\end{equation}
imposing no value on the information $\Y_{i,j}(a_i)$ of sensor $v_i$
if other sensor has better measurements on $r_j$.
The second option is to employ the function
\begin{equation}
 W_j(a) = h\Big(\sum_{v_i \in \V_j(a)}W_{i,j}(a_i)\Big)
\label{eqn:total_reward2}
\end{equation}
for a monotonically increasing concave function $h$ with $h(0) = 0$.
This function weakly accepts the value of the measurement
which is not the best among the sensors.

The goal of this paper is to present a cooperative/distributed
action selection algorithm leading
vision sensors $\V$ to a joint action $a$ maximizing the global objective function defined by
\begin{equation}
 W(a) = \sum_{r_j \in \R}W_j(a)
\label{eqn:social_welfare}
\end{equation}
%
under the constraint that
$I^{\rm info}_{i,j}(a_i)$ and $I^{\rm qual}_{i,j}(a_i)$ are available only after 
$v_i$ executes an action $a_i$.

\subsection{Examples of Function $I^{\rm info}_{i,j}$}
\label{sec:3.2}

\begin{figure}
\begin{center}
\begin{minipage}[b]{5cm}
\includegraphics[height=3.3cm]{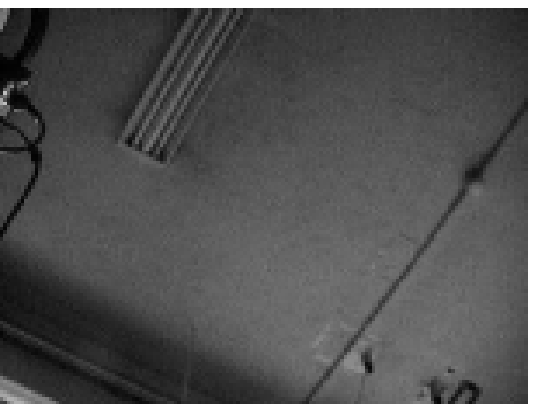}
\caption{Initial image}
\label{fig:5}
\end{minipage}
\hspace{.2cm}
\begin{minipage}[b]{5cm}
\includegraphics[height=3.3cm]{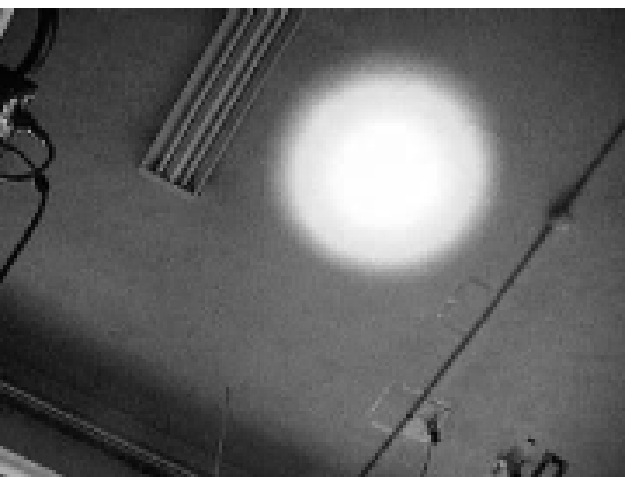}
\caption{Current image}
\label{fig:6}
\end{minipage}
\hspace{.2cm}
\begin{minipage}[b]{5cm}
\includegraphics[height=3.3cm]{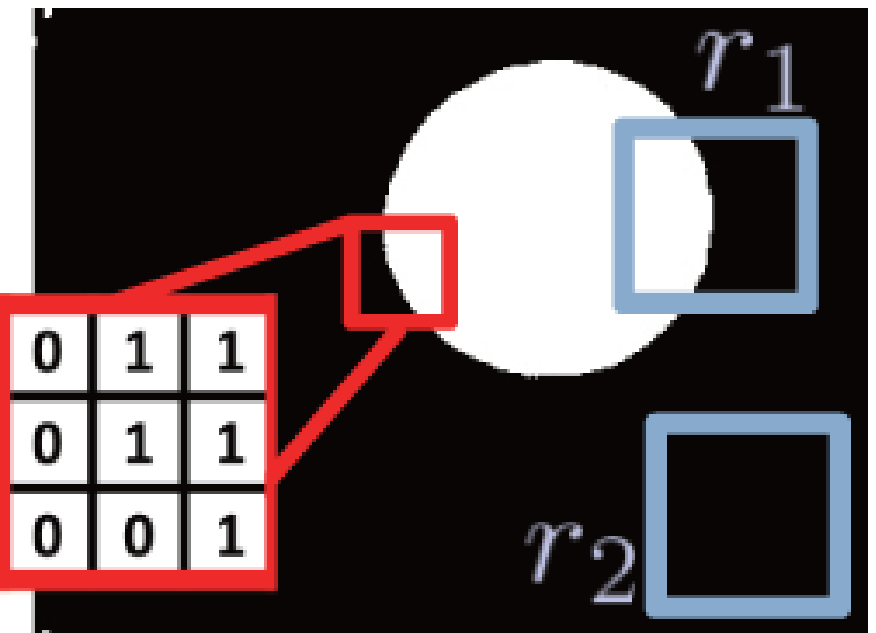}
\caption{Outputs of function $I^{\rm ec}_{i,l}$}
\label{fig:7}
\end{minipage}
\end{center}
\end{figure}

In this subsection, we will introduce examples of 
the function $I^{\rm info}_{i,j}(a_i)$.

We first consider a normally static environment,
and suppose that we are interested only in 
whether or not each pixel $s_{i,l}$ captures
environmental changes.
The requirement is reflected by 
\begin{equation}
I^{\rm info}_{i,j}(a_i) = \sum_{l \in {\mathcal F}_{i,j}(a_i)} I^{\rm ec}_{i,l}(a_i),\
I^{\rm ec}_{i,l}(a_i) = \left\{
\begin{array}{ll}
0,&\mbox{if } \|y_{i,l} - y^0_{i,l}\|\leq \bar{y}_{\rm threshold}\\
1,&\mbox{if } \|y_{i,l} - y^0_{i,l}\|> \bar{y}_{\rm threshold}
\end{array}
\right.\ l \in {\mathcal F}_{i,j}(a_i),
\label{eqn:info1}
\end{equation}
where $\Y^0_{i,j}(a_i) = \{y^0_{i,l}\}_{l \in {\mathcal F}_{i,j}(a_i)}$
is the stored initial image and $\bar{y}_{\rm threshold}$
is a positive scalar.
A small $I^{\rm info}_{i,j}(a_i)$ means that
no serious event occurs at around $r_j$, while
a large $I^{\rm info}_{i,j}(a_i)$ indicates some 
environmental changes.
For example, suppose that the initial image for an action $a_i$ is given by
the gray scale image in Fig. \ref{fig:5},
and the current measurement with the same $a_i$ is the image in Fig. \ref{fig:6},
where light is shined only in a part of the image.
Then, the outputs of the function $I^{\rm ec}_{i,l}$ with $\bar{y}_{\rm threshold} = 20$
are given as Fig. \ref{fig:7}, where the black and white areas
correspond to $I^{\rm ec}_{i,l} = 0$ and $I^{\rm ec}_{i,l} = 1$, respectively.
If two resources $r_1$ and $r_2$ are captured as in 
Fig. \ref{fig:7}, then $r_1$ including environmental changes
must provide a larger $I^{\rm info}_{i,j}(a_i)$ than $r_2$.

We next consider the situation where a visual sensor network
monitors the sky to help prediction/estimation of
the solar radiation via remote sensing from a satellite \cite{TNHHTPN_JGR11}.
Then, the image data of both the bright blue sky
and the cloud contain little information since such information
can be provided by the low resolution data from a satellite.
Namely, it is desirable for vision sensors to provide images
capturing the borders between blue and cloudy sky.
A metric to measure such amount of information is the image entropy \cite{ZRG_CDC11}.
For example, let us assume that a resource provides Fig. \ref{fig:8},
and the other resource provides Fig. \ref{fig:9}.
Then, the image entropy of Image 1 after a gray-scale processing is equal to 
$7.506$
while that of Image 2 is $6.269$.
As expected, Image 1 containing both of the blue sky and 
the cloud provides a larger entropy.

The other option is to run some existing cloud detection algorithm as in \cite{GUCSCK_AMTD12} 
and to count the number of pixels corresponding to the corner as $I^{\rm info}_{i,j}$.
The output of \cite{GUCSCK_AMTD12} is illustrated by red dots in Figs. \ref{fig:10} and \ref{fig:11},
and Figs. \ref{fig:10} and \ref{fig:11} provide $I^{\rm info}_{i,j} = 1557$
and  $I^{\rm info}_{i,j} = 63$, respectively.

\begin{figure}
\begin{center}
\begin{minipage}{8cm}
\begin{center}
\includegraphics[width=5cm, height=3.375cm]{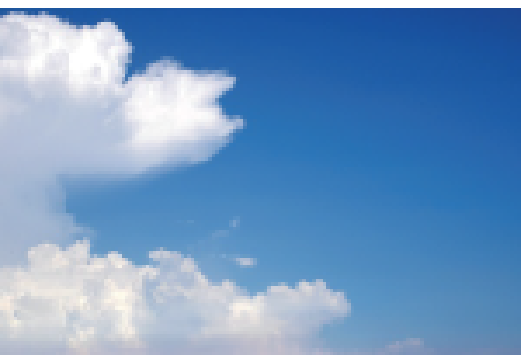}
\caption{Image 1 (entropy $= 7.506$)}
\label{fig:8}
\end{center}
\end{minipage}
\begin{minipage}{8cm}
\begin{center}
\includegraphics[width=5cm, height=3.375cm]{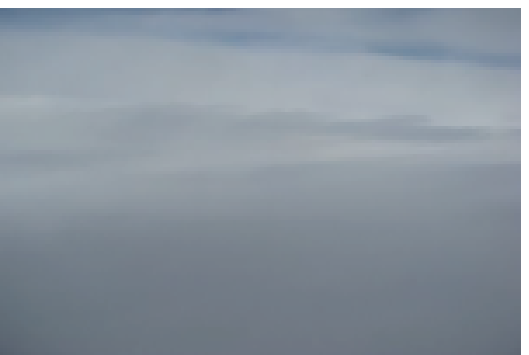}
\caption{Image 2 (entropy $= 6.269$)}
\label{fig:9}
\end{center}
\end{minipage}
\medskip

\begin{minipage}{8cm}
\begin{center}
\includegraphics[width=5cm, height=3.375cm]{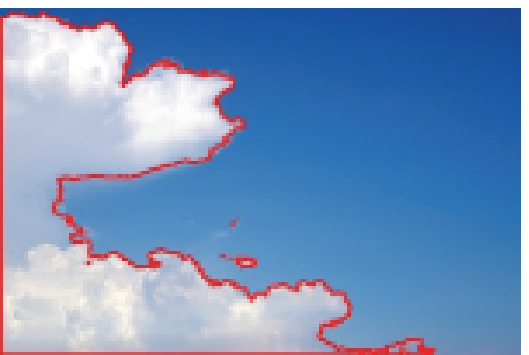}
\caption{Corner detection for Image 1 ($I^{\rm info}_{i,j} = 1557$)}
\label{fig:10}
\end{center}
\end{minipage}
\begin{minipage}{8cm}
\begin{center}
\includegraphics[width=5cm, height=3.375cm]{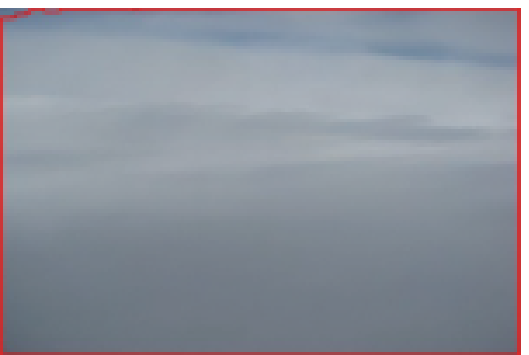}
\caption{Corner detection for Image 2 ($I^{\rm info}_{i,j} = 63$)}
\label{fig:11}
\end{center}
\end{minipage}
\end{center}
\end{figure}

\subsection{Utility Design and Potential Games}
\label{sec:3.3}

We next design a utility function 
$U_i(a)$ which vision sensor $v_i \in \V$ basically tries to maximize.
Here, we use the 
marginal contribution utility \cite{MAS_SMC09}, \cite{GMW_PER10}
for the global objective 
(\ref{eqn:social_welfare}) as
\begin{eqnarray}
U_i(a) = W(a) - W^{-i}(a),
\label{eqn:utility}
\end{eqnarray}
where $W^{-i}(a)$ is
equal to the global objective $W$ 
in the case that $v_i$ views no polygon and the other sensors 
take actions $a_{-i}$.
Then, collecting all factors $\V$, $\A$, 
$\{\C_i(\cdot)\}_{v_i\in \V}$, and the utility functions
$\{U_i(\cdot)\}_{v_i\in \V}$ in (\ref{eqn:utility}),
we can define a {\it constrained strategic game}
\begin{equation}
\Gamma = (\V, \A, \{U_i(\cdot)\}_{v_i\in \V},\{\C_i(\cdot)\}_{v_i\in \V}).
\label{eqn:strategiv_game}
\end{equation}

We next introduce the following terminologies.
\begin{definition}[Constrained Potential Games \cite{MAS_SMC09}, \cite{ZM_SIAM13}] \label{def:2}
{\rm A constrained strategic game $\Gamma$
is said to be a {\it constrained potential game} with 
potential function $\phi :{\mathcal A}\rightarrow {\mathbb R}$
if for all $v_i\in \V$,
every $a_{i}\in {\mathcal A}_i$
and every $a_{-i} \in \prod_{i'\neq i}\A_{i'}$,
the following equation holds 
for every $a'_i\in \C_i(a_i)$.}
\begin{equation}
U_i(a'_i,a_{-i})-U_i(a_i,a_{-i})=\phi (a'_i,a_{-i})-\phi (a_i,a_{-i})
\label{eqn:potential_game}
\end{equation}
\end{definition}
\begin{definition}[Constrained Nash Equillibria \cite{MAS_SMC09}, \cite{ZM_SIAM13}]
\label{def:3}
{\rm
For a constrained strategic game $\Gamma$,
a joint action $a^*\in \A$ is said to be a {\it constrained
pure Nash equilibrium} if
the equation $U_i(a^*_i,a^*_{-i})=\max_{a_i\in \C_i(a_i^*)}U_i(a_i,a^*_{-i})$
holds for all $v_i \in \V$.
%
%
}
\end{definition}
Then, it is well known that any constrained 
potential game has at least one Nash equilibrium
and the potential function maximizers must be contained in the set
of Nash equilibria \cite{MAS_SMC09}, \cite{ZM_SIAM13}.
In addition, we have the following lemma
from the feature of the marginal contribution utility.
%
\begin{lemma}\cite{GMW_PER10}
\label{lem:1}
The strategic game $\Gamma$ in (\ref{eqn:strategiv_game}) with (\ref{eqn:utility})
constitutes a constrained potential game with potential function $\phi$ equal to
the global objective function $W$.
\end{lemma}

In the remaining part of this section, we clarify a
computation procedure of the utility function $U_i(a)$
after a joint action $a$ is determined.
The quantities $I_{i,j}^{\rm info}$ and $I_{i,j}^{\rm qual}$
for any $r_j \in \R_i(a_i)$
must be locally computed 
since they evaluate the image information $\Y_{i,j}(a_i)$ itself.
From (\ref{eqn:reward}), $W_{i,j}(a_i)$
is also locally computable at $v_i\in \V$.
In addition, we can prove the following lemma.
\begin{lemma}
\label{lem:2}
Under Assumption \ref{ass:2},
the utility function $U_i(a)$ in (\ref{eqn:utility}) for any fixed $a \in \A$
is uniquely determined for all $v_i\in \V$ if 
the information 
\begin{equation*}
(I^{\rm com}_{i'})_{v_{i'} \in \N_i},\ I^{\rm com}_{i'}
:= (W_{i',j}(a_{i'}),j)_{r_j \in \R_{i'}(a_{i'})}
\end{equation*}
is available for vision sensor $v_i$.
\end{lemma}
\begin{proof}
We first define $W_{j}^{-i}(a)$ which is equal to
the value of $W_j(a)$ when  $v_i$ views 
no polygon and the other sensors 
take actions $a_{-i}$.
Then, Equation (\ref{eqn:utility}) implies that
\begin{eqnarray}
U_i(a) =
\sum_{r_j\in \R_i(a_i)}W_j(a) - W_{j}^{-i}(a)
\label{eqn:dist_utility}
\end{eqnarray}
since $W_j(a),\ r_j \notin \R_i(a_i)$ is independent 
of $W_{i,j}(a_i),\ v_i \notin \V_j(a)$ from (\ref{eqn:total_reward}).
(\ref{eqn:total_reward}) and (\ref{eqn:dist_utility}) also mean
$U_i$ 
is determined by $\{W_{i',j}(a_{i'})\}_{v_{i'}\in \V_j(a), r_j \in \R_i(a_i)}$.
Assumption \ref{ass:2} implies that
$\cup_{r_j \in \R_i(a_i)}\V_j(a)$ must be included in
$\N_i$ for any $a\in \A$, which completes
the proof.
\end{proof}
Lemma \ref{lem:2} and the knowledge of $\tilde{W}_j$ 
mean that $U_i$
is computable in a distributed fashion in the sense of graph $G$.
More importantly, if $v_i$ just needs to feedback $U_i(a)$ for a fixed joint action $a$,
he has only to locally execute the image processing, 
which is in general the hardest process in the monitoring task,
and to exchange the compacted information $I^{\rm com}_i$ through communication.

Hereafter, we use the following assumption, which is not
restrictive since
it is satisfied by just scaling the global objective function $W$ appropriately. 
\begin{assumption}\label{ass:3}{\rm 
For any $(a, a')$ satisfying 
$a'_i\in \C_i(a_i)$ and $a_{-i}=a'_{-i}$, the inequality
$U_i(a')-U_i(a) < 1/2$ holds
for all $v_i \in \V$.}
\end{assumption}

\section{Learning Algorithm}
\label{sec:4}

Since the potential function $\phi$ is equal to the global
objective function $W$ (Lemma \ref{lem:1}), 
the only remaining task is to design an action selection 
rule determining $a_i(k)$ at each round 
$k\in {\mathbb Z}_+ := \{0, 1, 2, \cdots\}$ such that
the joint action $a(k)$ is eventually led to 
the potential function maximizers.
Note that due to the constraint that
$I^{\rm info}_{i,j}(a_i)$ is available only after 
an action $a_i$ is executed and the communication constraints specified by graph $G$,
$a_i(k)$ must be determined based on the past actions $\{a_i(k')\}_{k'\leq k-1}$,
visual measurements $\{(y_{il}(k'))_{l \in {\mathcal S}_i}\}_{k'\leq k-1}$
and communication messages $\{(I^{\rm com}_{i'}(k'))_{v_{i'} \in \N_i}\}_{k'\leq k}$
from neighbors in $\N_i$.
Now, we see from Lemma \ref{lem:2} that
an algorithm determining $a_i(k)$ based on the past actions and utilities
$\{a_i(k'), U_i(a(k'))\}_{k'\leq k-1}$ meets the requirement,
and such algorithms are called {\it payoff-based learning}
\cite{MS08_GEB12}, \cite{ZM_SIAM13}.

In this paper, we present a learning algorithm
{\it Payoff-based Inhomogeneous Partially Irrational Play (PIPIP)}
based on the algorithm in \cite{ZM_SIAM13}.
In the algorithm, 
every vision sensor $v_i\in \V$ chooses his own action $a_i(k)$ {\it concurrently}
at each round $k\in {\mathbb Z}_+$ using only the past two
actions $a_i(k-2), a_i(k-1)$ and
utilities $U_i(a(k-2)), U_i(a(k-1))$
stored in memory based on the policies called {\it exploration},
{\it exploitation} and {\it irrational decision}.

Initially, each sensor $v_i \in \V$ executes an action $a_i$ 
randomly (uniformly) chosen from $\A_i$ and feedbacks the resulting utility $U_i(a)$.
Then, set $a_i(0) = a_i(1) = a_i$ and $U_i(a(0)) = U_i(a(1)) = U_i(a)$.
At round $k\geq 2$, if $U_i(a(k-1))\geq U_i(a(k-2))$ holds, then
every sensor $v_i \in \V$ chooses action $a_i(k)$ concurrently according to the rule:
\begin{itemize}
\item ({\it exploration}) $a_i(k)$ is randomly chosen from $\C_i(a_i(k-1))\setminus \{a_i(k-1)\}$
with probability $\varepsilon$, 
\item ({\it exploitation}) $a_i(k)=a_i(k-1)$ with probability $1-\varepsilon$,
\end{itemize}
where the parameter $\ve \in (0, 1/2]$ is called an {\it exploration rate}. 
Otherwise ($U_i(a(k-1))< U_i(a(k-2))$), action $a_i(k)$ is chosen
according to the rule:
\begin{itemize}
\item ({\it exploration}) $a_i(k)$ is randomly chosen from $\C_i(a_i(k-1))
\setminus \{a_i(k-1), a_i(k-2)\}$  with probability $\varepsilon$,
\item ({\it exploitation}) $a_i(k)=a_i(k-2)$ with probability 
$(1-\varepsilon)(1-\kappa\varepsilon^{\Delta_i(k)})$,
where $\Delta_i(k) := U_i(a(k-2))-U_i(a(k-1))$ and $\kappa \in [0, 1/2]$,
\item ({\it irrational decision}) $a_i(k) = a_i(k-1)$
with probability $(1-\varepsilon)\kappa \varepsilon^{\Delta_i(k)}$.
\end{itemize}
It is clear under the third item of Assumption \ref{ass:1} that 
the action $a_i(k)$ is well-defined.

\begin{algorithm}[t]
\caption{Payoff-based Inhomogeneous Partially Irrational Play}
\label{alg:1}
\begin{algorithmic}
\STATE {\bf Initialization:}
Action $a_i$ is chosen randomly from $\A_i$.
Set $a_i^1 \leftarrow a_i,\  a_i^2\leftarrow a_i,\
U_i^1\leftarrow U_i(a),\ U_i^2 \leftarrow U_i(a)$,
$\Delta_i \leftarrow 0$ for all $v_i\in \V$ and $k \leftarrow 2$.
\STATE {\bf Step 1:} Update $\ve$, if necessary. 
\STATE {\bf Step 2:} 
If $U_i^1\geq U_i^2$, then 
\begin{equation*}
 a_i^{\rm tmp} \leftarrow \left\{
\begin{array}{l}
{\rm rnd}(\C_i(a_i^1)\setminus \{a_i^1\}), \mbox{ with probability } 
\varepsilon\\
a_i^1,\mbox{ with probability }1-\varepsilon
\end{array}
\right..
\end{equation*}
Otherwise, 
\begin{equation*}
 a_i^{\rm tmp} \leftarrow \left\{
\begin{array}{l}
{\rm rnd}(\C_i(a_i^1)\setminus \{a_i^1,a_i^2\}), \mbox{ with probability } 
\varepsilon\\
a_i^1,\mbox{ with probability }(1-\varepsilon)(\kappa\cdot \varepsilon^{\Delta_i})\\
a_i^2,\mbox{ with probability }
(1-\varepsilon)(1-\kappa\cdot \varepsilon^{\Delta_i})
\end{array}
\right..
\end{equation*}
\STATE {\bf Step 3:} Execute the selected action $a_i^{\rm tmp}$.
\STATE {\bf Step 4:} Compute Utility $U_i(a^{\rm tmp})$ by the following procedure:
\STATE\ \  {\bf Step 4.1:} Feedback the visual measurements $\{y_{i,l}\}_{s_l \in {\mathcal S}_i}$, extract
$(I^{\rm info}_{i,j}, I^{\rm qual}_{i,j})$ from the

\ \ \ \ \ \ \ \ \ \ \ \ \ \ measurements, and calculate $W_{i,j}$
for all $r_j \in \R_i(a_i^{\rm tmp})$.
\STATE\ \  {\bf Step 4.2:} Set $I^{\rm com}_i \leftarrow (W_{i,j}(a_i^{\rm tmp}), j)_{r_j \in \R_i(a_i^{\rm tmp})}$
and send it to $\N_i$. 
\STATE\ \  {\bf Step 4.3:} Receive $(I^{\rm com}_{i'})_{v_{i'} \in \N_i}$ and compute
utility $U_i(a^{\rm tmp})$.
\STATE {\bf Step 5:} Set
$a_i^2 \leftarrow a^1_i,\ a_i^1 \leftarrow a_i^{\rm tmp},\ U_i^2 \leftarrow U_i^1,\
U_i^1 \leftarrow U_i(a^{\rm tmp})$ and 
$\Delta_i \leftarrow U_i^2 - U_i^1$.
\STATE {\bf Step 6:} 
$k \leftarrow k+1$ and go to {\bf Step 1}. 
\end{algorithmic}
\end{algorithm}

Finally, each $v_i$ executes the selected action $a_i(k)$ and computes
the resulting utility $U_i(a(k))$ 
by the procedure stated at the end of Section \ref{sec:3}.
At the next round, sensors repeat the same procedure.

The procedure of PIPIP associated with sensor $v_i \in \V$ 
is compactly described in Algorithm \ref{alg:1},
where the function ${\rm rnd}({\mathcal A}')$ outputs an action
chosen from the set ${\mathcal A}'$ according to the uniform distribution.
The important feature of PIPIP is to allow
vision sensors to make the irrational decisions at {\bf Step 2}.
Indeed, Algorithm \ref{alg:1} with $\kappa = 0$
is the same as the algorithm in \cite{ZM_SIAM13}.

Let us first consider Algorithm \ref{alg:1} 
with a constant $\varepsilon \in (0,1/2]$
skipping {\bf Step 1}, which is called 
{\it Payoff-based Homogeneous Partially Irrational Play (PHPIP)} in this 
paper.
Then, we have the following theorem, which will be proved
in the next section.
\begin{theorem}
\label{thm:2}
{\rm 
Consider a constrained strategic game 
$\Gamma$ in (\ref{eqn:strategiv_game}) with (\ref{eqn:utility})
satisfying Assumptions \ref{ass:1} and \ref{ass:3},
and sensors following PHPIP.
Then, given any probability $p<1$, 
if the exploration rate $\varepsilon$ is sufficiently small, 
for all sufficiently large $k$, the following equation holds.
\begin{equation}
{\rm Prob} \left[a(k)\in \arg\max_{a\in \A} \phi(a)\right] > 
 p.
 \label{eqn:prob_pve}
\end{equation}}
\end{theorem}
Theorem \ref{thm:2} ensures that 
the optimal actions maximizing the global objective $W$
are eventually selected with high probability (Lemma \ref{lem:1})
if the exploration rate $\varepsilon$ is sufficiently small.
However, it is difficult to reveal a quantitative relation between the probability $p$
and $\ve$ in (\ref{eqn:prob_pve}).



We next consider Algorithm \ref{alg:1} with
the following update rule of $\varepsilon$ at {\bf Step 1}
similarly to \cite{ZM_SIAM13}.
\begin{equation}
\varepsilon(k)=k^{-\frac{1}{n(D+1)}},
\label{eqn:epsilon}
\end{equation}
where $D$ is defined as $D:={\max}_{v_i\in \V} D_i$ and
$D_i$ is the minimal number of steps required for
transitioning between any two actions of  $v_i$.
Then, the following theorem holds. 
\begin{theorem}\label{thm:1}{\rm 
Consider a constrained strategic game 
$\Gamma$ in (\ref{eqn:strategiv_game}) with (\ref{eqn:utility})
satisfying Assumptions \ref{ass:1} and \ref{ass:3}.
Suppose that each vision sensor obeys Algorithm \ref{alg:1}
with (\ref{eqn:epsilon}).
Then, if the parameter $\kappa$ is chosen so as to satisfy
\begin{equation}
\kappa\in \Big(\frac{1}{C - 1}, \frac{1}{2}\Big],\ 
C := \max_{v_i\in \V}\max_{a_i\in {\mathcal A}_i}|\C_i(a_i)|,
\label{eqn:kappa}
\end{equation}
the following equation holds.
\begin{equation}
\lim_{k\rightarrow \infty}{\rm Prob}\left[a(k)\in 
\arg\max_{a\in \A} \phi(a)\right]=1
\label{eqn:convergence}
\end{equation}
%
}
\end{theorem}
From Lemma \ref{lem:1},
Equation (\ref{eqn:convergence}) means that the probability that 
vision sensors take 
one of the global objective function maximizers converges to $1$.

The statement of Theorem \ref{thm:2} is compatible with \cite{MS08_GEB12}.
The contribution of PIPIP is simplicity of the action selection rule
and the convergence result in Theorem \ref{thm:1} similarly to \cite{ZM_SIAM13}.
Meanwhile, the concurrent version of the algorithm 
in \cite{ZM_SIAM13} ensures convergence in probability to
Nash equilibria but does not guarantee convergence to potential function
maximizers i.e. the global objective function maximizers
in the context of this paper.
In contrast, thanks to the irrational decisions, PIPIP leads sensors 
to the potential function maximizers.
Namely, PIPIP embodies the desirable features of both algorithms
in \cite{MS08_GEB12} and \cite{ZM_SIAM13}.

Remark that PIPIP guarantees the above theorems
even in the presence of the action constraints
specified by the set $\C_i$.
This allows one to take account of physical constraints
of PTZ cameras.
More importantly, constraints can be useful as a design parameter.
Indeed, depending on the scenarios,
persistent possibility of explorations toward all elements of $\A_i$
can lead to volatile behavior of the global objective function
in the practical use of the learning algorithm.
In such a situation, adding some virtual constraints works
for stabilizing the evolution.

Note that both of the above theorems address static games,
which indicates that the present algorithm works in the scenarios of monitoring
normally static environment including sudden changes 
since the environment before/after the change are both static.
In addition, application of the conclusions in \cite{SM_12}
to Theorem \ref{thm:2} means that the present algorithm
successfully adapt to the gradual environmental changes
as long as the speed of the dynamics is sufficiently slow.



\section{Stability Analysis}
\label{sec:5}


\subsection{Preliminary: Fundamentals of Resistance Tree}
\label{sec:5.1}

Let us consider a Markov process $\{P_k^0\}$ defined over a finite state space $\X$.
A perturbed Markov process 
$\{P_k^\varepsilon \},\ \varepsilon \in [0, 1]$ is defined 
as a process such that the transition of $\{P_k^\varepsilon \}$ 
follows $\{P_k^0\}$ with probability $1 - \varepsilon$
and does not follow with probability $\varepsilon$.
In particular, we focus on a {\it regular perturbation} defined below.

\begin{definition}[Regular Perturbation \cite{Y_BK01}]
\label{def:4}
{\rm
A family of stochastic processes $\{P^\varepsilon_k\}$ is called a {\it 
 regular perturbation} of $\{P_k^0\}$ if the
following conditions are satisfied:\\
{\bf (A1)} For some $\varepsilon^* >0$, the process $\{P^\varepsilon _k\}$ is 
irreducible and aperiodic for all $\varepsilon \in (0,\varepsilon^*]$.\\
{\bf (A2)} Let us denote by $P_{x^1x^2}^\varepsilon$
the transition probability from 
$x^1 \in \X$ to $x^2 \in \X$ along with the Markov process $\{P^\varepsilon _k\}$.
Then, 
${\rm lim}_{\varepsilon \rightarrow 0}P_{x^1x^2}^\varepsilon = P_{x^1x^2}^0$
holds for all $x^1, x^2\in \X$.\\
{\bf (A3)}
If $P_{x^1x^2}^\varepsilon >0$ for some $\varepsilon$, 
then there exists a real number $\chi (x^1\rightarrow x^2)\geq 0$ such that 
\begin{equation}
{\rm lim}_{\varepsilon \rightarrow 0}
\frac{P_{x^1x^2}^\varepsilon}{\varepsilon ^{\chi (x^1\rightarrow x^2)}}
\in (0,\infty ),
\label{eq:2.3}
\end{equation}
where $\chi (x^1\rightarrow x^2)$ is called {\it resistance of transition} 
 $x^1 \rightarrow x^2$.
}
\end{definition}

We next consider a path $\rho$ from $x \in \X$ to $x' \in \X$ 
along with transitions 
$x^{(1)} = x \rightarrow x^{(2)}\rightarrow\cdots\rightarrow x^{(m)} = x'$
and denote by $P^\varepsilon(\rho)$ the probability of the sequence of transitions.
Then, {\it resistance $\chi(\rho)$ of a path}
$\rho$ is defined as the value satisfying
\begin{equation}
\lim_{\varepsilon\rightarrow 0} \frac{P^\varepsilon(\rho)}{\varepsilon^{\chi(\rho)}} \in (0, \infty).
\label{eqn:resistance_path}
\end{equation}
Then, it is easy to confirm that $\chi(\rho)$ is simply given by
\begin{equation}
 \chi(\rho) = \sum_{i = 1}^{m-1}\chi(x^{(i)}\rightarrow x^{(i+1)}).
\label{eqn:resistance_path1}
\end{equation}

A state $x \in \X$ 
is said to communicate with state $x' \in \X$ if both $x \rightsquigarrow x'$ 
and $x' \rightsquigarrow x$ hold, where the notation $x \rightsquigarrow x'$ 
implies that $x'$ is accessible from $x$ i.e.
a process starting at state $x$ has non-zero probability of 
transitioning into $x'$ at some point. 
A {\it recurrent communication class} is a class such that 
every pair of states in the class communicates with each other and
no state outside the class is accessible from the class.
Let $H_1,\cdots ,H_J$ be recurrent communication classes 
 of unperturbed Markov process $\{P_k^0\}$.
Then, within each class, there is a path with zero resistance from every state to every other.
In the case of a perturbed Markov process $\{P_k^\varepsilon\}$,
there may exist several paths from states in $H_{l}$ to states in $H_{l'}$
for any two distinct classes $H_{l}$ and $H_{l'}$.

We next define a weighted directed graph $G_R = (\H, \E_R, \W_R)$ 
over the recurrent communication classes $\H = \{H_1,\cdots,H_J\}$ of $\{P_k^0\}$,
where the weight $w_{ll'}\in \W_R$ of each edge $(H_{l},H_{l'})$
is equal to the minimal resistance among all paths over $\{P_k^{\ve}\}$ 
from a state in $H_{l}$ to a state in $H_{l'}$.
We also define {\it $l$-tree} which is 
a spanning tree over $G_R$ with root $H_l$ such that, for every $H_{l'} \neq H_l$, 
there is a unique path from $H_{l'}$ to $H_l$.
The {\it resistance of an $l$-tree} is the sum of
the weights of all the edges over the tree.
The tree with the minimal resistance among all $l$-trees 
is called the {\it minimal resistance tree}, 
and the corresponding minimal resistance is called
{\it stochastic potential} of $H_l$.
Let us now introduce the 
notion of {\it stochastically stable state}.
\begin{definition}[Stochastically Stable State \cite{Y_BK01}]
\label{def:5}
{\rm
A state $x\in \X$ is said to be stochastically stable, if $x$ satisfies 
${\lim}_{\varepsilon\rightarrow 0+}\mu_x(\varepsilon)>0$, where
$\mu_x(\varepsilon)$ is the value of an element of 
stationary distribution $\mu(\varepsilon)$ corresponding to state $x$.}
\end{definition}
Then, we can use the following result linking
stochastically stable states and stochastic potential.
\begin{lemma}{\rm \cite{Y_BK01}}
\label{lem:3}
{\rm 
Let $\{P^\varepsilon_k\}$ be a regular perturbation of $\{P_k^0\}$. 
Then ${\lim}_{\varepsilon\rightarrow 0+}\mu (\varepsilon)$ exists and 
the limiting distribution is a stationary distribution of $\{P_k^0\}$. 
Moreover the stochastically stable states are contained 
in recurrent communication classes of $\{P_k^0\}$ with minimum stochastic potential.}
\end{lemma}

\subsection{Auxiliary Results}
\label{sec:5.2}

We first consider PHPIP with a constant exploration rate $\varepsilon$.
Then, the transitions of the state $z(k) = (a(k-1),a(k))$ for PHPIP are 
described by a perturbed homogeneous Markov process $\{P^\varepsilon\}$ 
on the state space 
%
${\mathcal B}:=\{(a,a')\in {\mathcal A}\times {\mathcal A}|\ a'_i\in 
\C_i(a_i)\ \forall i \in {\mathcal V}\}$.
In the following, 
we use the notation 
$\diag(\A') = \{(a,a)\in \A \times \A|\ a\in {\mathcal A}'\},\ \A' \subseteq \A$
similarly to \cite{ZM_SIAM13}.
%
\begin{figure}[t]
\begin{center}
\begin{minipage}[t]{8cm}
\begin{center}
\includegraphics[width=7.8cm]{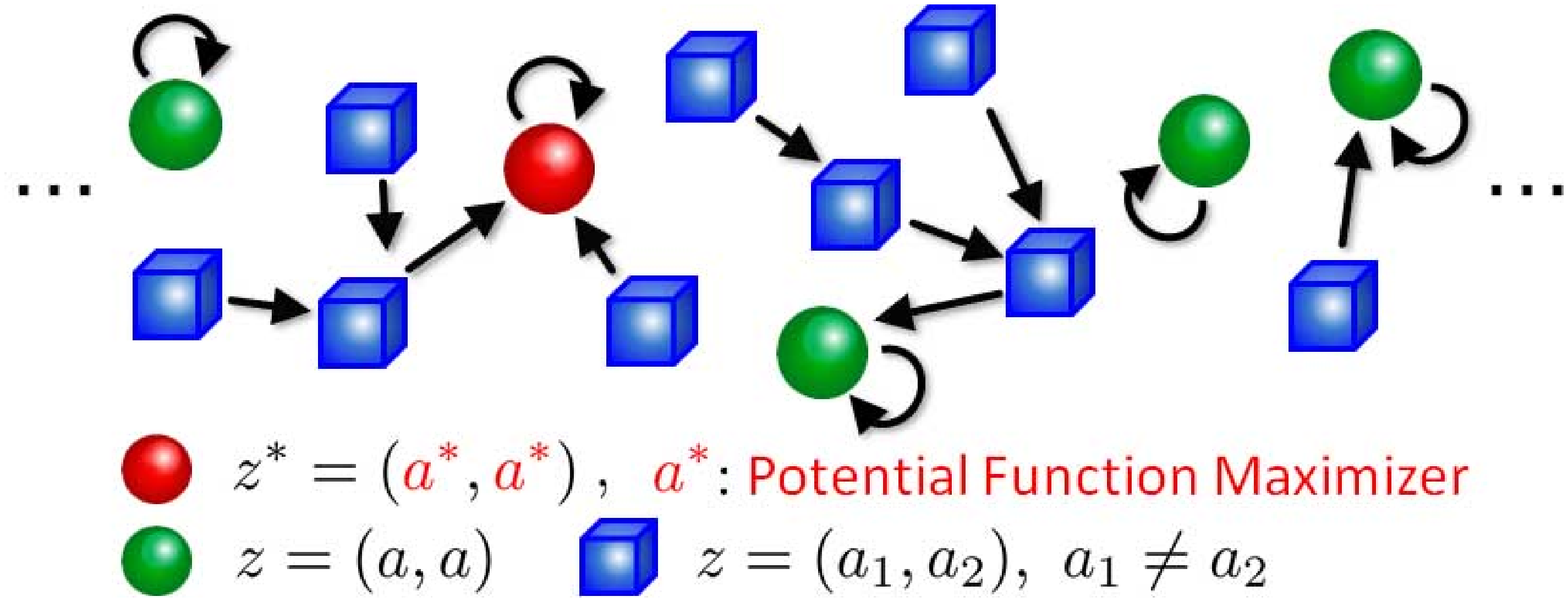}
\caption{Image of Markov process $\{P^{0}\}$}
\label{fig:12}
\end{center}
\hspace{.5cm}
\end{minipage}
\begin{minipage}[t]{8cm}
\begin{center}
\includegraphics[width=7.8cm]{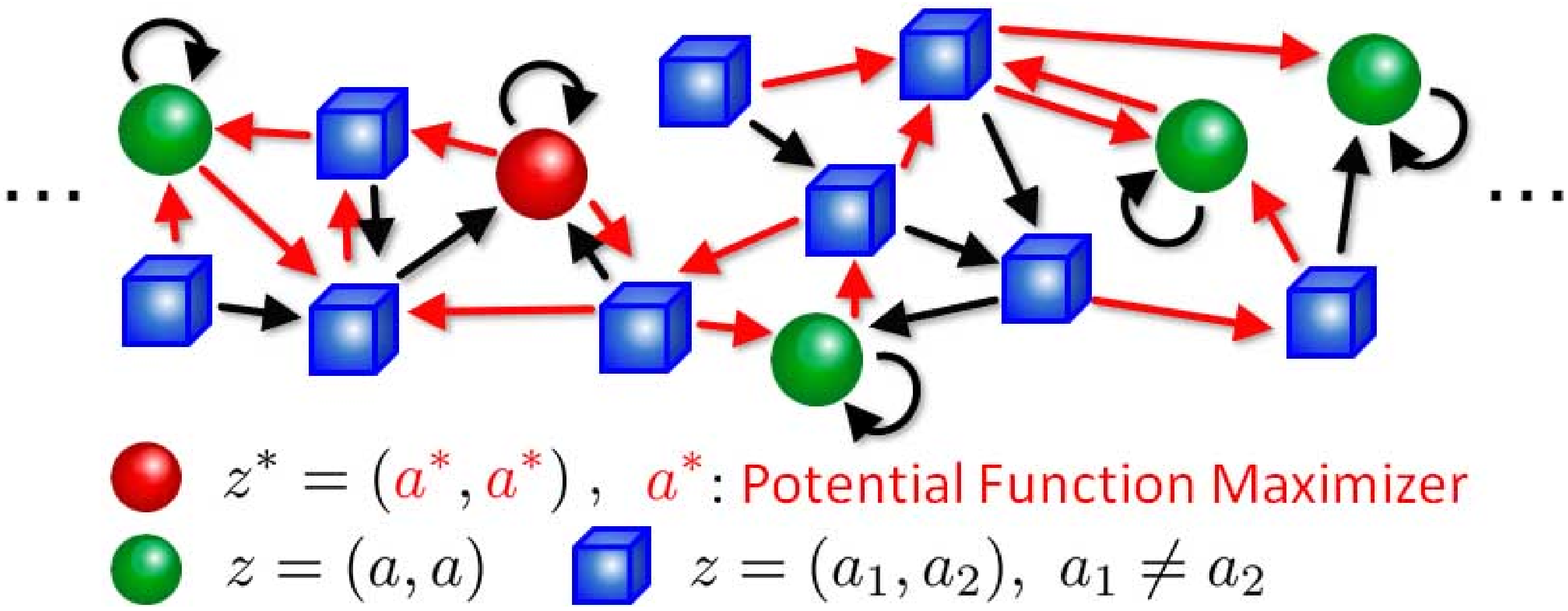}
\caption{Image of Markov process $\{P^{\ve}\}$ 
(Transitions colored by red can happen only when $\ve \neq 0$)}
\label{fig:13}
\end{center}
\end{minipage}
\end{center}
\end{figure}

In terms of the Markov process $\{P^\varepsilon\}$ induced by PHPIP,
the following lemma holds.
\begin{lemma}\label{lem:A1}
{\rm 
The Markov process $\{P^\varepsilon\}$ induced by PHPIP applied to
the constrained strategic game $\Gamma$ in (\ref{eqn:strategiv_game}) with (\ref{eqn:utility}) is 
a regular perturbation of $\{P^0 \}$ under Assumption \ref{ass:1}.}
\end{lemma}
%
\begin{proof}
See Appendix \ref{app:1}
\end{proof}
From Lemma \ref{lem:A1}, the perturbed process
$\{P^\varepsilon\}$ is irreducible and aperiodic, and hence there exists 
a unique stationary distribution $\mu (\varepsilon)$ for every 
$\varepsilon$. 
We also see from the former half of Lemma \ref{lem:3} that 
$\lim_{\varepsilon\rightarrow 0+}\mu (\varepsilon)$ exists and 
the limiting distribution is the stationary distribution of $\{P^0\}$.

We also have the following lemma. 
\begin{lemma}
\label{lem:A2}
{\rm 
Consider the Markov process $\{P^\varepsilon\}$ induced by PHPIP
applied to the constrained strategic game $\Gamma$ in  (\ref{eqn:strategiv_game}) with (\ref{eqn:utility}).
Then, the recurrent communication classes 
of the unperturbed Markov process $\{P^0 \}$ are given by elements of 
$\diag(\A)=\{(a,a)\in \A \times \A|\ a\in {\mathcal A}\}$, namely 
\begin{equation}
 H_{\iota} = \{(a^{\iota},a^{\iota})\},\ \iota =  1,\cdots,|\A|.
\label{eqn:recurrent_com_class}
\end{equation}}
\end{lemma}
%
\begin{proof}
See Appendix \ref{app:2}.
\end{proof}
The lemma means that all the paths over the Markov process $\{P^{0}\}$
eventually reach and remain at a state such that $a(k-1) = a(k)$
as illustrated in Fig. \ref{fig:12}.
Meanwhile, the process $\{P^{\ve}\},\ \ve\neq 0$ contains 
paths traversing two of such states as in 
Fig. \ref{fig:13}.

We next introduce the following terminology,
where we use the notation
\begin{eqnarray}
&&\E_{s} := \{(z=(a,a),z'=(a',a'))\in \diag(\A)\times\diag(\A)|\ 
\nonumber\\
&&\hspace{3cm}\exists v_i \in \V \mbox{ s.t. }a_i\in \C_i(a_i'), a_i \neq a_i' \mbox{ and } a_{-i} 
= a_{-i}'\}.
\label{eqn:set_es}
\end{eqnarray}
%
\begin{definition}[Straight Route]
\label{def:7}
{\rm
A feasible path over $\{P^\varepsilon\}$ 
from $z^1=(a^1,a^1)$ to $z^2=(a^2,a^2)$ such that
$(z^1,z^2)\in \E_{s}$
is said to be a {\it straight route} 
if the path describes the two rounds transitions that only 
$v_i$ satisfying $a^1_i \neq a^2_i$ chooses
$a_i(k+1) = a_i^2$ through exploration at the first round
and he also chooses $a_i(k+2) = a_i^2$ at the second round
while the other sensors do not update their actions (Fig. \ref{fig:14}).
In addition, a feasible path over $\{P^\varepsilon\}$ 
from $z^1\in \diag(\A)$ to $z^2\in \diag(\A)$ 
is said to be an {\it $M$-straight-route} 
if the path contains $M$ nodes in $\diag(\A)$ 
including $z^1$ and $z^2$, visits the $M$ nodes only once,
and any path between such nodes are straight routes.}
\end{definition}

\begin{figure}
\begin{center}
\includegraphics[width=7cm]{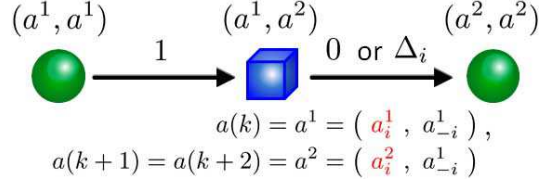}
\caption{Straight route (the numbers above the arrows
describe resistances of transitions)}
\label{fig:14}
\end{center}
\end{figure}

In terms of the straight route, we have the following lemmas.
\begin{lemma}
\label{lem:A3}
{\rm 
Consider paths from any state $z^1=(a^1,a^1)\in \diag(\A)$ to any state 
$z^2=(a^2,a^2)\in \diag(\A)$ such that $(z^1,z^2) \in \E_{s}$
over the process $\{P^\varepsilon\}$ induced by PHPIP applied to
the game $\Gamma$ in (\ref{eqn:strategiv_game}) with (\ref{eqn:utility}).
Then, under Assumption \ref{ass:3}, the resistance $\chi(\rho)$ of 
the straight route $\rho$ from $z^1$ to $z^2$ is strictly 
smaller than $3/2$ and $\chi(\rho)$ is minimal
among all paths from $z^1$ to $z^2$.}
\end{lemma}
%
\begin{proof}
See Appendix \ref{app:3}.
\end{proof}
%

\begin{lemma}\label{lem:A4}
{\rm 
Consider the Markov process $\{P^\varepsilon\}$ induced by PHPIP
applied to the game $\Gamma$ in (\ref{eqn:strategiv_game}) with (\ref{eqn:utility}).
Let us describe an $M$-straight-route $\rho$ 
from state $z^1 = (a^1,a^1)$ to state 
$z^2=(a^2,a^2)$ as
%
$\rho:\ z^{(1)} = z^1{\Rightarrow}z^{(2)}{\Rightarrow}
\cdots \Rightarrow z^{(M-1)}{\Rightarrow}
z^{(M)} = z^2$, 
%
where $z^{(l)} = (a^{(l)},a^{(l)})\in \diag(\A), l\in \{1,\cdots, M\}$ 
and all the arrows between them are straight routes.
In addition, we consider the (reverse)
$M$-straight-route $\rho':\ z^{(1)} = z^1{\Leftarrow}z^{(2)}{\Leftarrow}\cdots
{\Leftarrow}z^{(M-1)}{\Leftarrow}z^{(M)} = z^2$
from $z^2$ to $z^1$.
%
%
Then, under Assumption \ref{ass:3}, if $\phi(a^0) > \phi(a^1)$, 
the inequality $\chi(\rho) > \chi(\rho')$ holds true.}
\end{lemma}
%
\begin{proof}
See Appendix \ref{app:4}.
\end{proof}

\subsection{Proof of Theorems}
\label{sec:5.3}

\subsection*{Proof of Theorem \ref{thm:2}}

\begin{figure}[t]
\begin{center}
\begin{minipage}[t]{8.5cm}
\includegraphics[width=8.5cm]{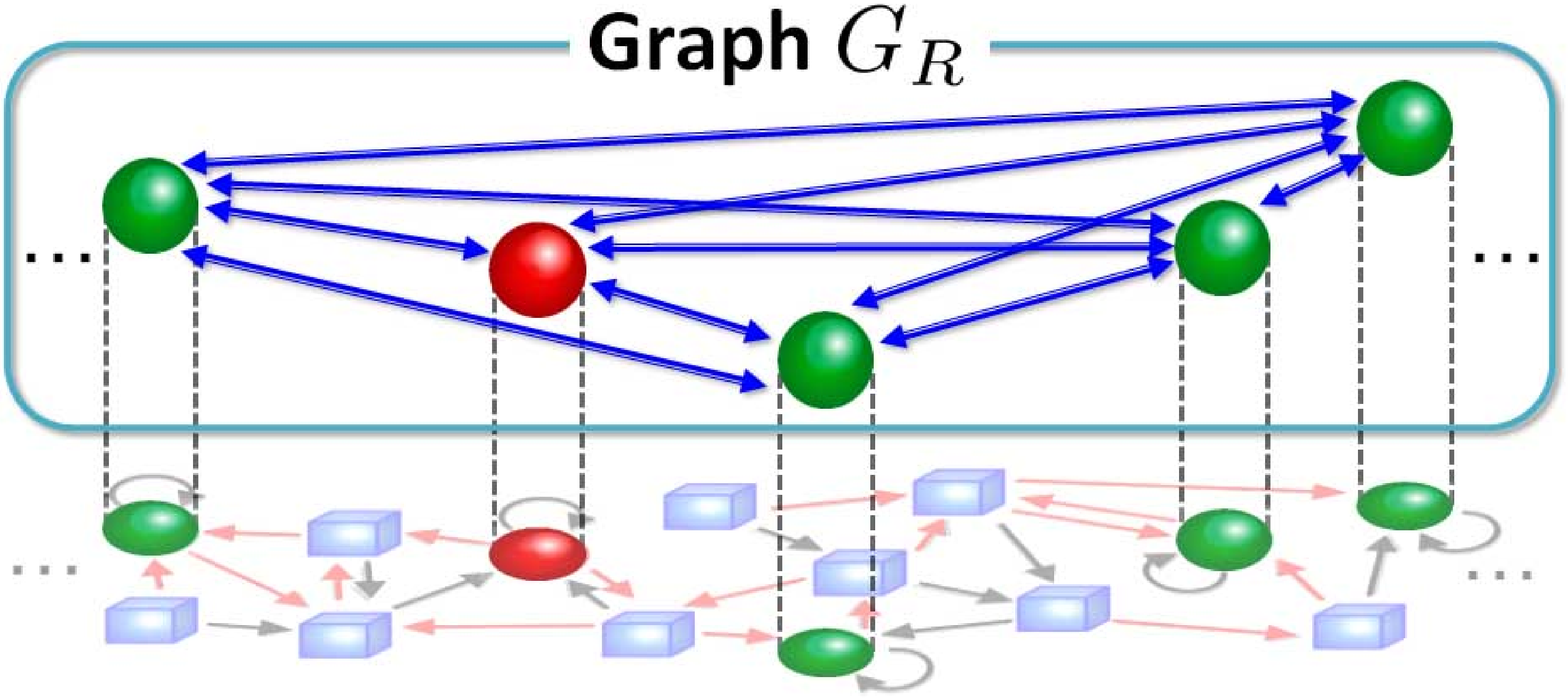}
\caption{Graph $G_R$}
\label{fig:15}
\end{minipage}
\hspace{.5cm}
\begin{minipage}[t]{7cm}
\includegraphics[width=7cm]{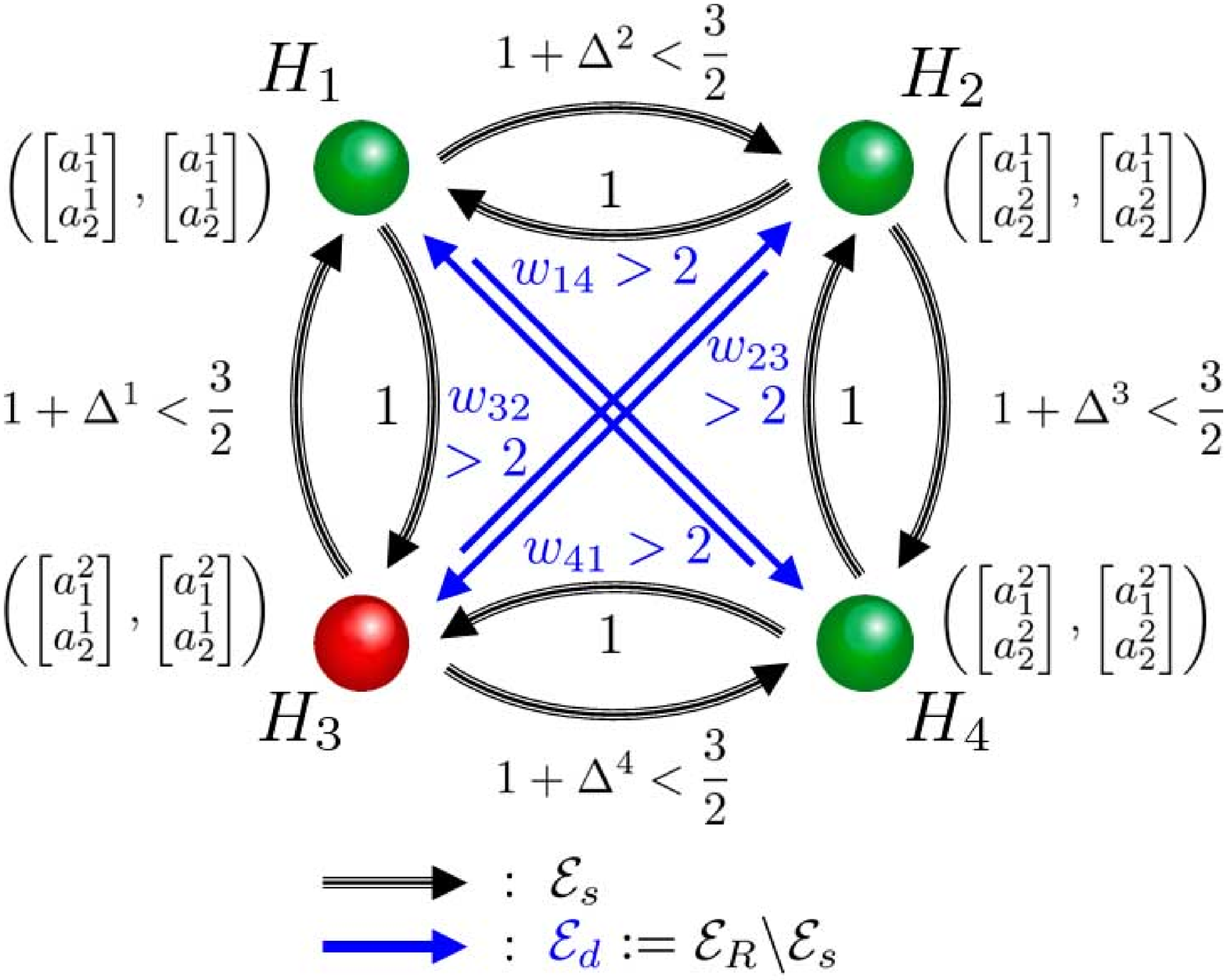}
\caption{Illustrative example of a two player game}
\label{fig:16}
\end{minipage}
\end{center}
\end{figure}

Let us form the directed graph $G_R = (\H, \E_R, \W_R)$ 
as in Subsection \ref{sec:5.1} over the recurrent communication classes
for the unperturbed Markov process $\{P^0\}$
induced by PHPIP (Fig. \ref{fig:15}).
From (\ref{eqn:recurrent_com_class}), the node set $\H$ of the graph $G_R$ is
given by $\diag(\A)$.
Since all the recurrent communication classes have only one element from (\ref{eqn:recurrent_com_class}),
the weight of the edge for any two states $z^1$ and $z^2 \in \diag(\A)$
is simply given by the path with the minimal resistance 
among all paths from $z^1$ to $z^2$ over $\{P^{\ve}\}$.
In addition, Lemma \ref{lem:3} proves that
if $(z^1,z^2)\in \E_{s}$ defined by (\ref{eqn:set_es}), the minimal weight 
is given by the straight route from $z^1$ to $z^2$.
For instance, let us consider a two player game
with $\A_1 = \{a_1^1, a_1^2\}$ and $\A_2 = \{a_2^1, a_2^2\}$.
Then, graph $G_R$ is illustrated as in Fig. \ref{fig:16},
where only the edges colored by blue are contained in $\E_d := \E_R\setminus \E_s$ and
have resistance greater than 2 since both of two players have to 
take exploration to escape from the state.


\begin{figure}
\centering
\begin{minipage}{7cm}
\includegraphics[width=7cm]{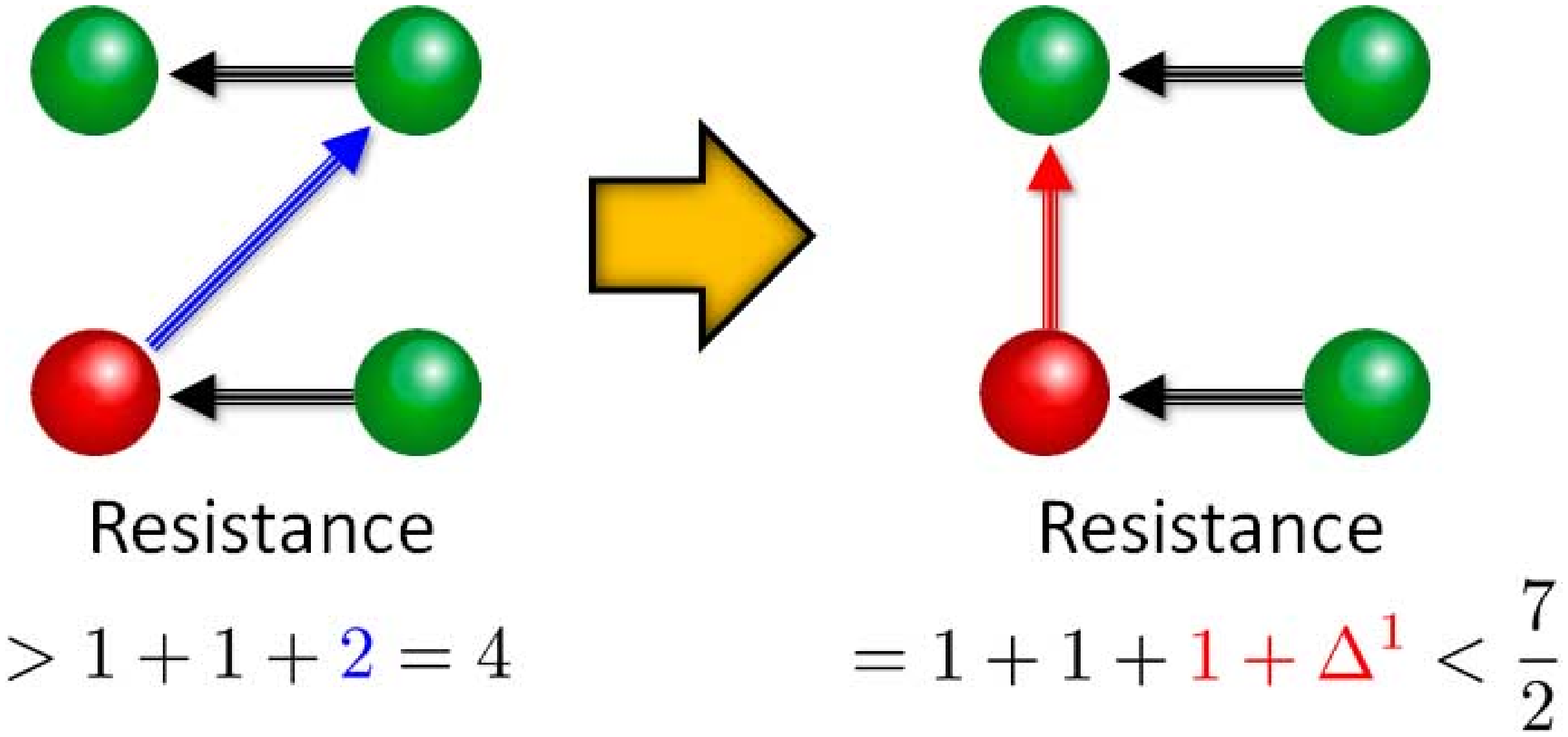}
\caption{Trees with root $H_1$ over graph in Fig. \ref{fig:16}}
\label{fig:17}
\end{minipage}
\hspace{.5cm}
\begin{minipage}{8cm}
\includegraphics[width=8cm]{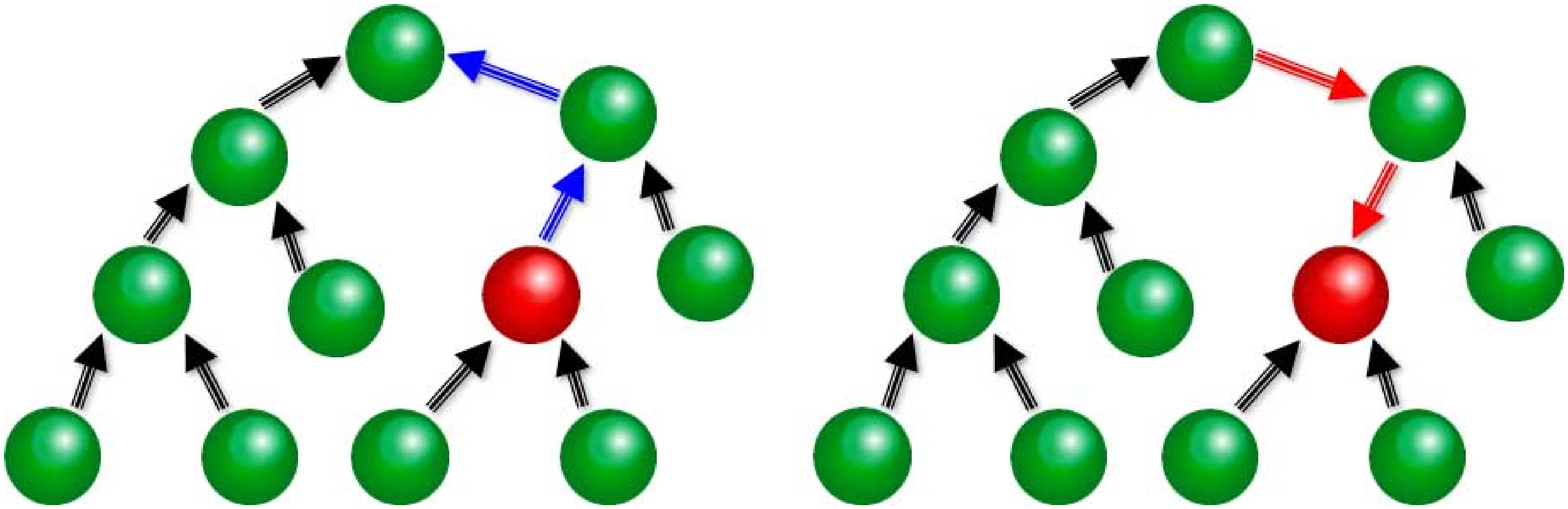}
\caption{Resistance trees (the red node is the state $z^* = (a^*, a^*)$ with 
the potential function maximizer $a^*$)}
\label{fig:18}
\end{minipage}
\end{figure}

Let us focus on $l$-trees over $G_R$ with a root
$H_l = z^l\in \diag(\A)$. 
Recall now that the resistance of the tree
is the sum of the weights of all edges
constituting the tree as defined in Subsection \ref{sec:5.1}.
Let us now consider a tree for the graph in Fig. \ref{fig:16}
containing an edge in $\E_d$ (Left figure of Fig. \ref{fig:17}).
Then, it is easy to confirm that a tree
with a smaller resistance can be formed by replacing the edge in $\E_d$
by an edge in $\E_s$ as illustrated in the right figure of Fig. \ref{fig:17}.
From this example, we have a conjecture that the minimal resistance tree 
consists only of edges in $\E_s$.
The following lemma proves that the conjecture is true 
for a general case with $n$ sensors.
\begin{lemma}
\label{lem:A5}
{\rm 
Consider the weighted directed graph $G_R$ constituted
from the Markov process $\{P^\varepsilon\}$
induced by PHPIP applied to
the constrained strategic game $\Gamma$ in (\ref{eqn:strategiv_game}) with (\ref{eqn:utility}).
Let us denote by $\T = (\diag(\A),\E_l,\W_l)$ the minimal resistance tree 
with root $z^l \in \diag(\A)$.
If Assumptions \ref{ass:1} and \ref{ass:3} are satisfied, 
then the edge set $\E_l$ must be a subset of $\E_{s}$.}
\end{lemma}
%
\begin{proof}
See Appendix \ref{app:5}.
\end{proof}
%

We are ready to prove Theorem \ref{thm:2}.
It is now sufficient to prove that all the 
stochastically stable states of $\{P^{\ve}\}$ 
are included in $\diag(\arg\max_{a\in \A}\phi(a))$,
since the probability of $a(k) \in \arg\max_{a\in \A}\phi(a)$ is 
greater than the probability of $z(k) = (a(k-1), a(k))\in \diag(\arg\max_{a\in \A}\phi(a))$.
We also see from (\ref{eqn:recurrent_com_class}) and Lemmas \ref{lem:3} and \ref{lem:A1}
that we need only to prove that the states in $\diag (\A)$
with the minimal stochastic potential 
are included in $\arg\max_{a\in \A}\phi(a)$.

We first introduce the notations $z' = (a',a')\in \diag (\A)$ 
 with $a' \notin \arg\max_{a\in \A}\phi(a)$ 
 and $z^*=(a^*,a^*)\in \diag({\mathcal A})$ with $a^*\in 
 \arg\max_{a\in \A}\phi(a)$.
Let the minimal resistance tree for the state
$z'$ be denoted by $T$.
Then, there exists a unique path $\rho$
from $z^*$ to $z'$ over $T$.
From Lemma \ref{lem:A5}, the path $\rho$ corresponds to an $M$-straight-route for some $M$.
Now, we can build a tree $T'$ with root $z^*$ such that
only the path $\rho$ is replaced by its reverse path $\rho'$
(Fig. \ref{fig:18}).
Then, we have $\chi(\rho)>\chi(\rho')$
from Lemma \ref{lem:A4} since $\phi(a^*) > \phi(a')$.
Thus, the resistance of $T'$ is smaller than that of $T$
and the stochastic potential of $z^*$ is smaller than or equal to
the resistance of $T'$.
The statement holds regardless of the selection of $a'$.
This completes the proof.

\subsection*{Proof of Theorem \ref{thm:1}}

Let us next consider PIPIP with time-varying $\varepsilon(k)$ and
first prove strong ergodicity of the inhomogeneous Markov process
$\{P_k^\varepsilon\}$ induced by PIPIP. 
Here, a Markov process $\{P_k\}$ over a state space $\X$ 
is said to be {\it strongly ergodic} \cite{IM_BK76}
if there exists a stochastic vector $\mu ^*$ such 
that the following equation holds
for any distribution $\mu$ on $\X$ and time $k_0$.
\begin{equation}
{\rm lim}_{k\rightarrow \infty}\mu P(k_0,k) = \mu^*,\ 
P(k_0,k) := \prod^{k-1}_{k'=k_0}P_{k'},\  0 \leq k_0 < k.
\label{eqn:strong_ergodicity}
\end{equation}
If $\{P_k\}$ is strongly ergodic, the distribution $\mu$ converges 
to the unique distribution $\mu^*$ from any initial state.
Meanwhile, the process $\{P_k\}$ is said to be {\it weakly ergodic} 
\cite{IM_BK76} if the following equation holds for all
${x^1},x^2,x^3 \in \X$ and ${k_0}\in {\mathbb Z}_+$.
\begin{equation*}
\lim_{k\rightarrow \infty}(P_{x^1x^3}(k_0,k)-P_{x^2x^3}(k_0,k)) = 0
\end{equation*}
Here, we also use the following lemmas.
\begin{lemma}{\rm \cite{IM_BK76}}
\label{prop:2}
{\rm 
A Markov process $\{P_k\}$ is strongly ergodic if the following conditions hold:
{\bf (B1)} The Markov process $\{P_k\}$ is weakly ergodic.
{\bf (B2)} For each $k$, there exists a stochastic vector 
$\mu ^{k}$ on $\X$ such that $\mu ^{k}$ is the left
eigenvector of the transition matrix $P(k)$ with eigenvalue 1.
{\bf (B3)} The eigenvector $\mu^{k}$ in (B2) satisfies 
$\sum^{\infty}_{k=0}\sum_{x\in \X}| \mu^{k}_{x} - \mu^{k+1}_{x} | < \infty $.
Moreover, if $\mu ^{*}=\lim_{k\rightarrow \infty }\mu^{t}$, then $\mu ^{*}$ is 
the vector in (\ref{eqn:strong_ergodicity}).}
\end{lemma}
\begin{lemma}{\rm \cite{IM_BK76}}
\label{lem:we}
{\rm 
A Markov process $\{P_k\}$ is weakly ergodic if and only if
there is a strongly increasing sequence of positive numbers
$k_{\iota},\ \iota \in {\mathbb Z}_+$ such that
\[
\sum_{\iota = 0}^{\infty}\min_{x^1,x^2\in \X} \sum_{x\in \BB}
\min\{P_{x^1x}(k_{\iota},k_{\iota+1}),P_{x^2x}(k_{\iota},k_{\iota+1})\}= \infty
\]
}
\end{lemma}

We next prove strong ergodicity of $\{P_k^\varepsilon\}$.
Conditions {\bf (B2)}, {\bf (B3)} in Lemma \ref{prop:2} 
can be proved in the same way as \cite{ZM_SIAM13}.
We thus mention only Condition {\bf (B1)}.
Recall now that the probability of transition 
$z^1\rightarrow z^2$ is given by (\ref{eqn:A.1}).
Since $\varepsilon(k)$ is strictly decreasing, there is $k_0\geq 1$ 
 such that $k_0$ is the first round satisfying
\begin{equation}
(1-\varepsilon(k))(1-\kappa\varepsilon(k) ^{\Delta _i })\geq 
 \frac{\varepsilon(k)}{C -1},\
1- \varepsilon(k) \geq \frac{\varepsilon(k) ^{(1-\Delta_i)}}{\kappa(C -1)}.
\label{eqn:conv_ep}
\end{equation}
The existence of $\varepsilon$ satisfying (\ref{eqn:conv_ep}) 
is guaranteed by (\ref{eqn:kappa}).
For all $k \geq k_0$, we have
\begin{equation*}
P_{z^1z^2}^\varepsilon(k) \geq 
\left(\frac{\varepsilon(k)}{C -1}\right)^n.
\end{equation*}
%
We next define
$x_z = \arg\min_{x\in \BB}P_{xz}(k,k+D+1)$
for any $z \in \BB$.
Then, similarly to \cite{ZM_SIAM13}, 
\[
\min_{z^0 \in \BB} P_{z^0z}(k,k+D+1)\geq P^{\ve(k)}_{x_zz^1}\cdots  
P^{\ve(k+D-1)}_{z^{D-1}z^D}P^{\ve(k+D)}_{z^Dz}
\geq \Big(\frac{\varepsilon(k)}{C -1}
\Big)^{n(D+1)}.
\]
Following \cite{ZM_SIAM13} again,
we have the following inequality with 
$k_{\iota} = (D+1)\iota$ and $(D+1)\iota_0\geq k_0$.
\begin{eqnarray}
\sum_{\iota = 0}^{\infty}\min_{z^1,z^2\in \BB}\sum_{z \in \BB}\min\{P_{z^1z}(k_{\iota},k_{\iota+1}),P_{z^2,z}(k_{\iota},k_{\iota+1})\}
\geq |\BB| \sum_{\iota = \iota_0}^{\infty}
\Big(\frac{\varepsilon(k)}{C -1}
\Big)^{n(D+1)}\nonumber\\
= \frac{|\BB|}{(C-1)^{(D+1)n}}\sum_{\iota=\iota_0}^{\infty}
\frac{1}{(D+1)\iota} = \infty
\nonumber
\end{eqnarray}
This inequality and Lemma \ref{lem:we} prove {\bf (B1)} and hence strong ergodicity of $\{P_k^\varepsilon\}$.
%
Thus, the distribution $\mu (\varepsilon(k))$ 
converges to the unique distribution $\mu ^*$  from any initial state.
In addition, we also have $\mu ^*=\mu (0) = {\lim}_{\varepsilon\rightarrow 0}\mu(\varepsilon)$
from ${\lim}_{k\rightarrow \infty}\varepsilon(k) = 0$.
We have already proved in Theorem \ref{thm:2}
that any state $z$ satisfying $\mu_z(0) > 0$ must be included 
in $\diag(\arg\max_{a\in \A}\phi(a))$.
Hence, (\ref{eqn:convergence}) holds and the proof of Theorem \ref{thm:1}
is completed.
\begin{figure}[t]
\begin{center}
\begin{minipage}[b]{4cm}
\begin{center}
\includegraphics[width=4cm]{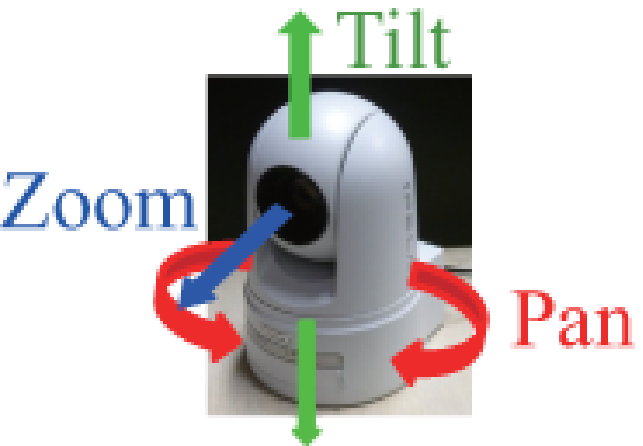}
\caption{PTZ network camera}
\label{fig:19}
\end{center}
\end{minipage}
\hspace{1.5cm}
\begin{minipage}[b]{8cm}
\begin{center}
\includegraphics[width=8cm]{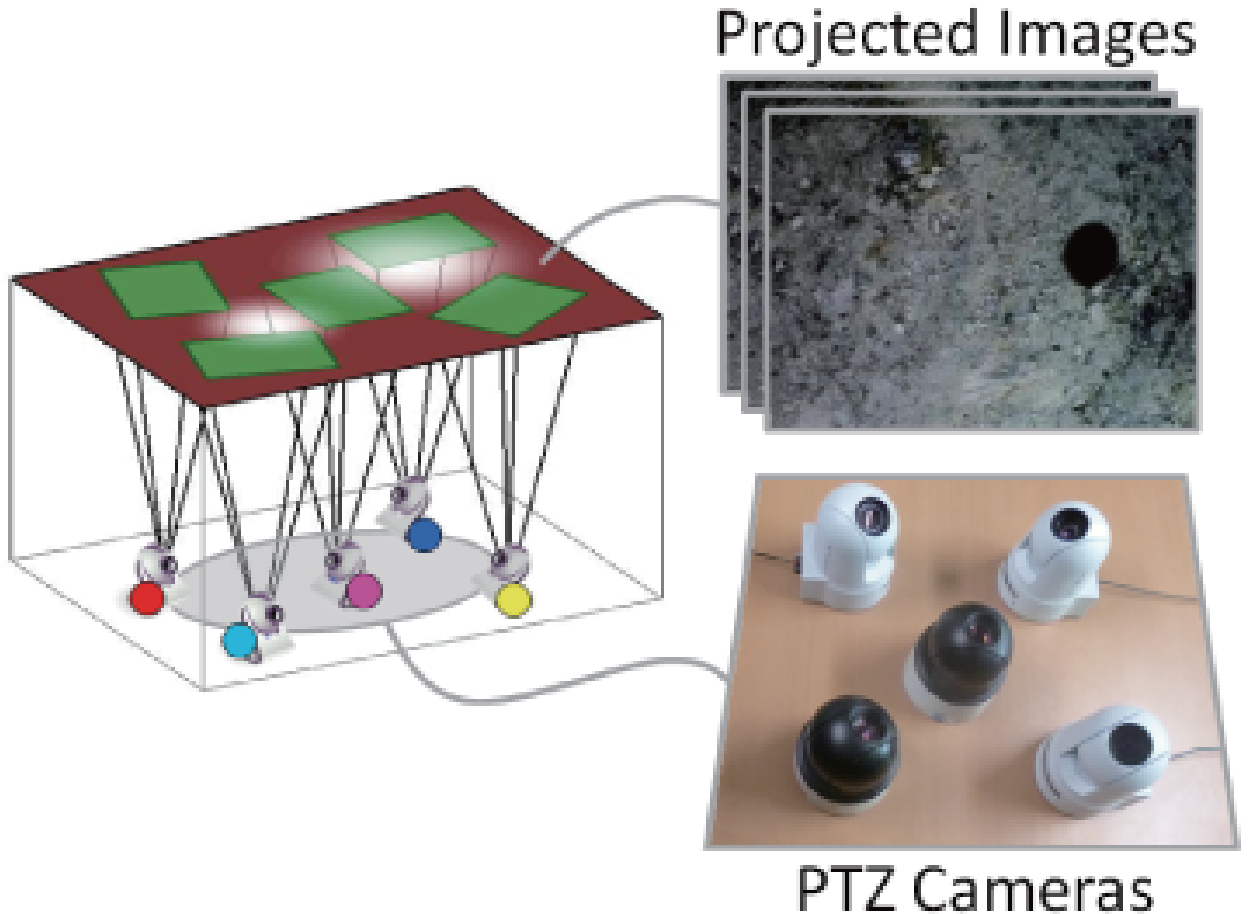}
\caption{Overview of scenario}
\label{fig:20}
\end{center}
\end{minipage}
\end{center}
\end{figure}

\section{Experimental Case Study}
\label{sec:6}

We finally demonstrate the effectiveness of the presented approach
through experiments on a testbed of PTZ visual sensor networks
consisting of 5 PTZ cameras $\V = \{v_1,v_2,v_3,v_4, v_5\}$ (Fig. \ref{fig:19}),
where two of them ($v_1$ and $v_2$) are IPELA SNC-EP520 (SONY Corp.) 
and the other three ($v_3$, $v_4$ and $v_5$) are IPELA SNC-RZ25N (SONY Corp.).
Note that the size of the acquired images is $640 \times 480$
($S_i \approx 3.0 \times 10^5$)
for every $v_i \in \V$.
In this experiment, all the algorithms including image processing 
and the learning algorithm are run via Visual C++ (Microsoft Corp.).

Let the five cameras monitor a ceiling of a room on which 
an image is projected as illustrated in Fig. \ref{fig:20}.
Namely, we regard the ceiling as the environment and 
divide it into 130 squares $\R = \{r_1,\cdots, r_{130}\}$, 10cm on a side.
Note that sensors $v_1,v_2,v_3,v_4, v_5$ are respectively marked by 
purple, yellow, cyan, blue and red circles.

\begin{figure}
\begin{center}
\begin{minipage}{5.4cm}
\begin{center}
\includegraphics[width=4cm,height=3cm]{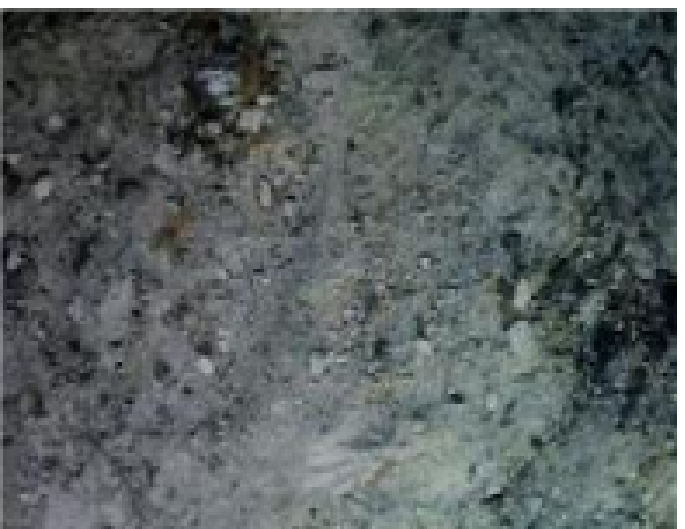}
\end{center}
\caption{Sample image}
\label{fig:21}
\end{minipage}
\begin{minipage}{5.4cm}
\begin{center}
\includegraphics[width=4cm,height=3cm]{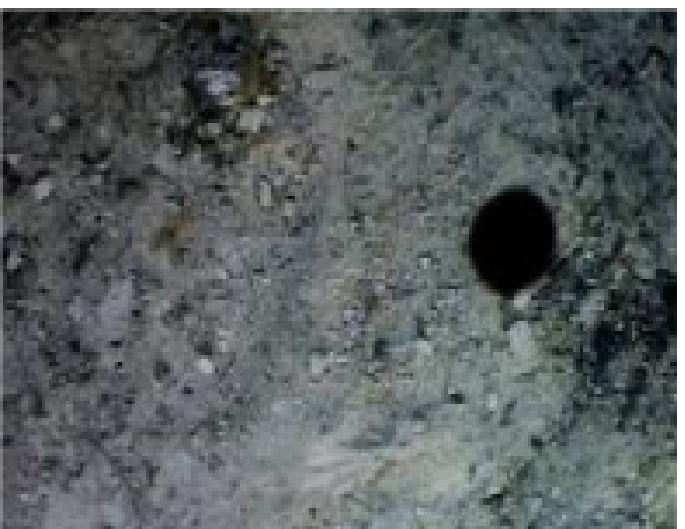}
\end{center}
\caption{Projected image until 1500 round}
\label{fig:22b}
\end{minipage}
\begin{minipage}{5.4cm}
\begin{center}
\includegraphics[width=4cm,height=3cm]{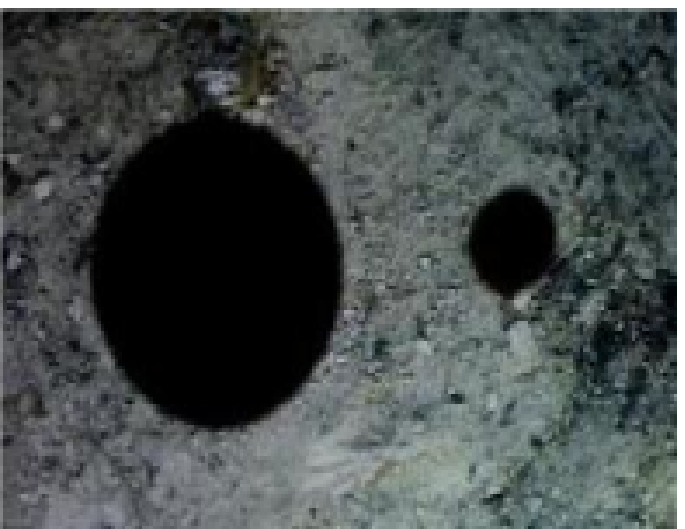}
\end{center}
\caption{Projected image after 1500 round}
\label{fig:22c}
\end{minipage}
\end{center}
\end{figure}

The action sets 
$\A_i=\varTheta_i\times \varPhi_i \times \varLambda_i$ are set as
\begin{eqnarray}
&&\varTheta_1 = \{(5/180)\pi n_{\theta}|\ n_{\theta} \in \{-34, -33, \cdots, 34\}\},\
\nonumber\\
&&\varTheta_i = \{(5/180)\pi n_{\theta}|\ n_{\theta} \in \{-2, -1, \cdots, 20\}\},\ i = 2,3,4,5,
\nonumber\\
&&\varPhi_i = \{(5/180)\pi n_{\varphi}|\ n_{\varphi} \in \{15, 16, \cdots, 18\}\},\ i \in \V,
\nonumber\\
&&\varLambda_1 = \varLambda_2 = \{6.8{\rm mm}, 13.6{\rm mm}\},\ 
\varLambda_3 = \varLambda_4 = \varLambda_5 = \{8.2{\rm mm}, 16.4{\rm mm}\}.
\nonumber
\end{eqnarray}
Just to stabilize the evolution of 
the objective function,
we introduce
the constrained action sets
\begin{eqnarray}
 \bar{A}_i(a_i) = \{(\theta_i', \varphi_i', \lambda_i')|\ 
|\theta_i' - \theta_i|\leq (5/180)\pi,\ 
|\varphi_i' - \varphi_i|\leq (5/180)\pi,\
\lambda_i \in \varLambda_i
\}
\nonumber
\end{eqnarray}
for all $a_i \in \A_i$ and $v_i\in \V$,
which clearly satisfies Assumption \ref{ass:1}.

\begin{figure}[t]
\begin{center}
\includegraphics[width=7cm]{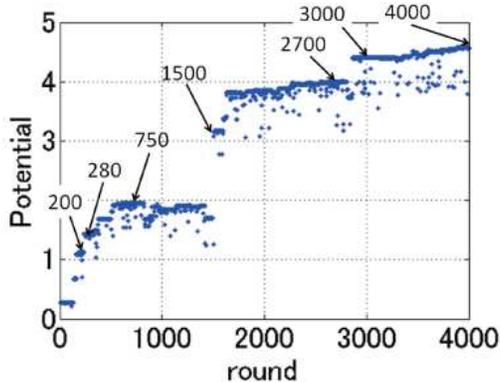}
\caption{Evolution of global objective function}
\label{fig:24}
\end{center}
\end{figure}
\begin{figure}[t]
\begin{center}
\subfigure[Initial State]{
\includegraphics[width=7.5cm,height=5cm]{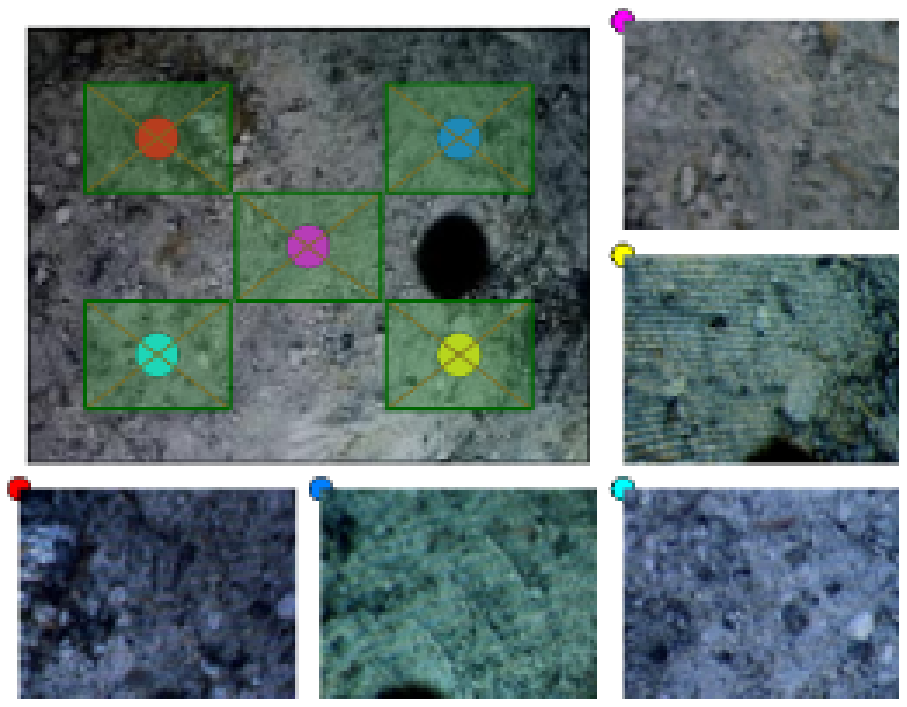}}
\hspace{.5cm}
\subfigure[200 round]{
\includegraphics[width=7.5cm,height=5cm]{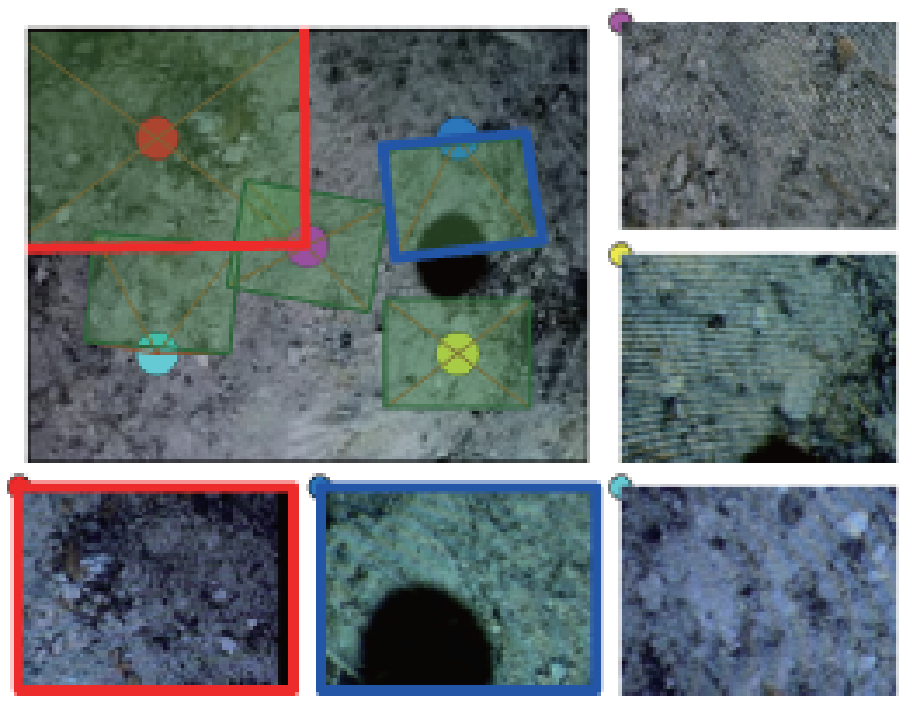}}

\subfigure[280 round]{
\includegraphics[width=7.5cm,height=5cm]{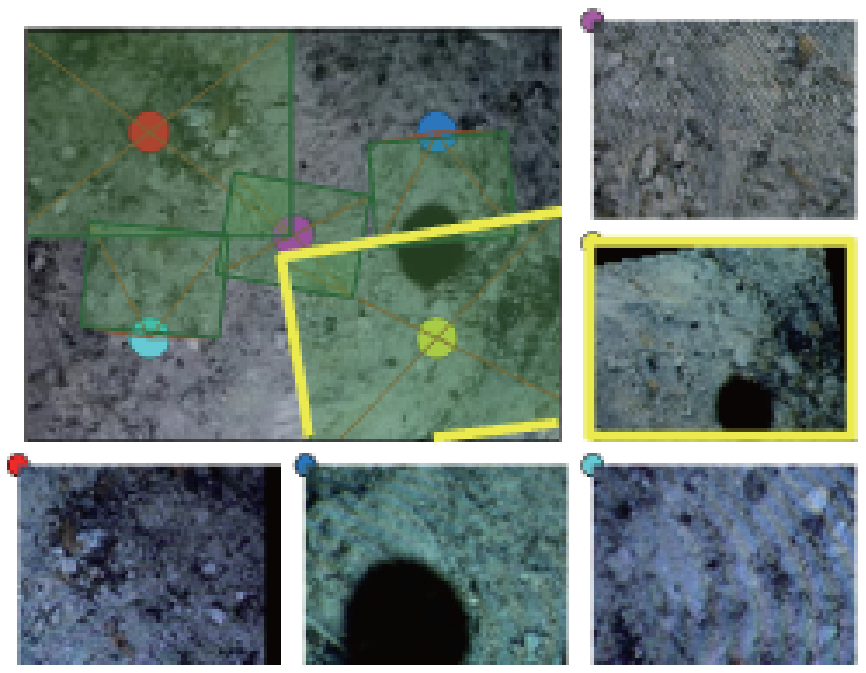}}
\hspace{.5cm}
\subfigure[750 round]{
\includegraphics[width=7.5cm,height=5cm]{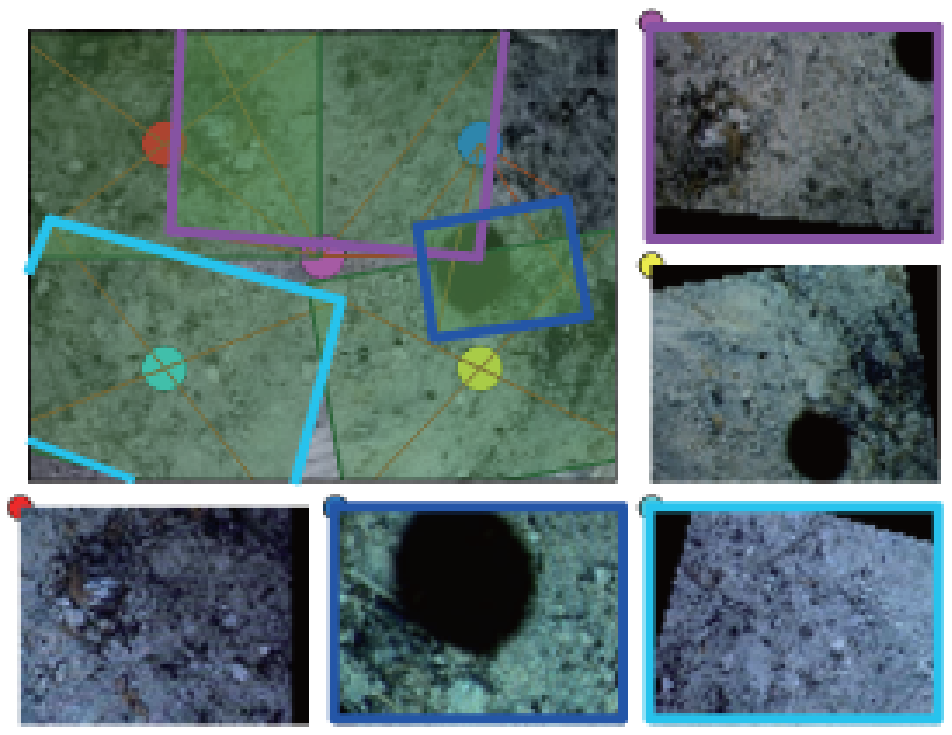}}
\caption{Snapshots of coverage area and acquired images for Fig. \ref{fig:22b}}
\label{fig:23a}
\end{center}
\end{figure}

The global objective function $W$ and utility function $U_i$
are selected as follows.
Suppose that 
each sensor $v_i$ stores in memory
a part of the sample image in Fig. \ref{fig:21} 
corresponding to each action in $\A_i$.
Then, we employ (\ref{eqn:info1}) as $I^{\rm info}_{i,j}(a_i)$.
Since (\ref{eqn:info1}) inherently embodies the function of 
$I^{\rm qual}_{i,j}(a_i)$, this experiment does not use $I^{\rm qual}_{i,j}(a_i)$,
and the function $\tilde{W}_{i,j}$ in (\ref{eqn:reward}) is chosen as
\begin{equation}
\tilde{W}_{i,j} = 
\left\{
\begin{array}{ll}
2I^{\rm info}_{i,j}(a_i), &
\mbox{if }r_j \in \R_i(a_i)\mbox{ and } 2I^{\rm info}_{i,j}(a_i) > \gamma\\
\gamma, & \mbox{if }r_j \in \R_i(a_i)\mbox{ and }
2I^{\rm info}_{i,j}(a_i) \leq \gamma.\\
0, & \mbox{if } r_j \not\in \R_i(a_i)
\end{array}
\right..
\label{eqn:reward_exp}
\end{equation}
The positive parameter $\gamma > 0$ is introduced
to place value of monitoring a region containing no useful information
in preparation to future environmental changes.
In this experiment, we set $\gamma = 1.5 \times 10^{-2}$.
We next define $W_j$ and $W$ 
by (\ref{eqn:total_reward1}) and (\ref{eqn:social_welfare}),
and then scale them so that Assumption \ref{ass:3} is satisfied.
Such a scaling is possible since the maximal value of $W_{i,j}$
in (\ref{eqn:reward_exp}) can be easily estimated.
Finally, the utility function $U_i$ is designed
according to (\ref{eqn:utility}).



In this experiment, we first project 
the image in Fig. \ref{fig:22b} on the environment,
which differs from the sample image in Fig. \ref{fig:21}
in that a small hole appears.
Then, we run the learning algorithm 
with $\varepsilon = 0.015$ and $\kappa = 0.120$ for 1500 rounds.
Note that all the sensors initially choose zoom-in mode (the larger $\lambda_i$).
After that, we change the image projected on the environment 
to the image in Fig. \ref{fig:22c} with a larger hole,
and leave the state for 2500 rounds. 
Due to the nature of the objective function,
it is intuitively desirable 
that sensors capture the holes 
with high resolution while keeping the total
coverage area as wide as possible.

\begin{figure}[t]
\begin{center}
\subfigure[1500 round]{
\includegraphics[width=7.5cm,height=5cm]{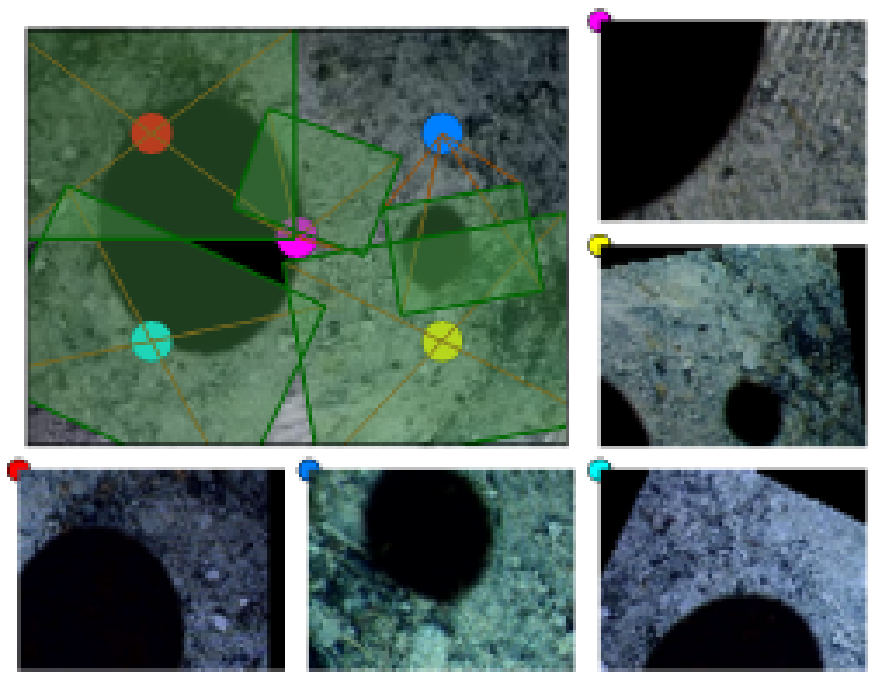}}
\hspace{.5cm}
\subfigure[2700 round]{
\includegraphics[width=7.5cm,height=5cm]{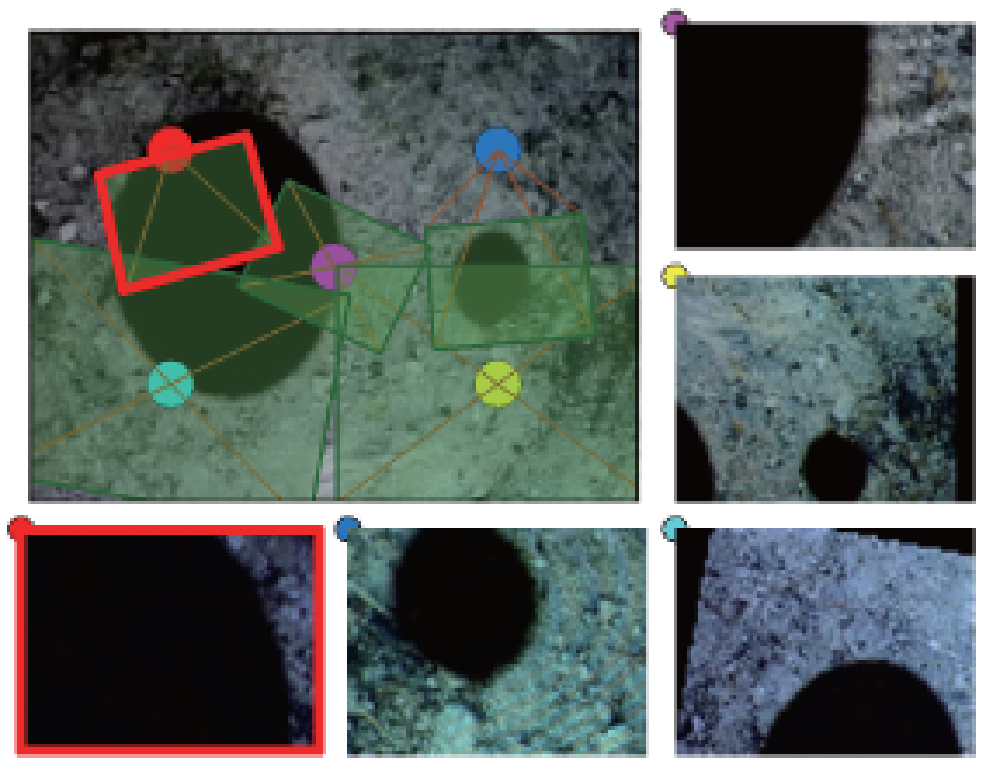}}

\subfigure[3000 round]{
\includegraphics[width=7.5cm,height=5cm]{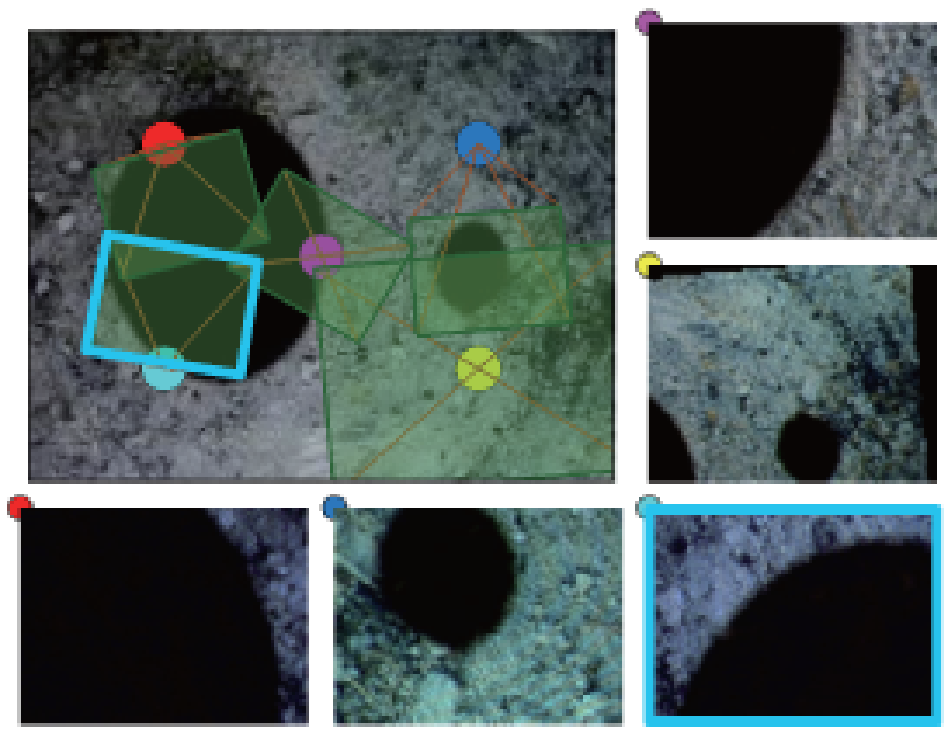}}
\hspace{.5cm}
\subfigure[4000 round]{
\includegraphics[width=7.5cm,height=5cm]{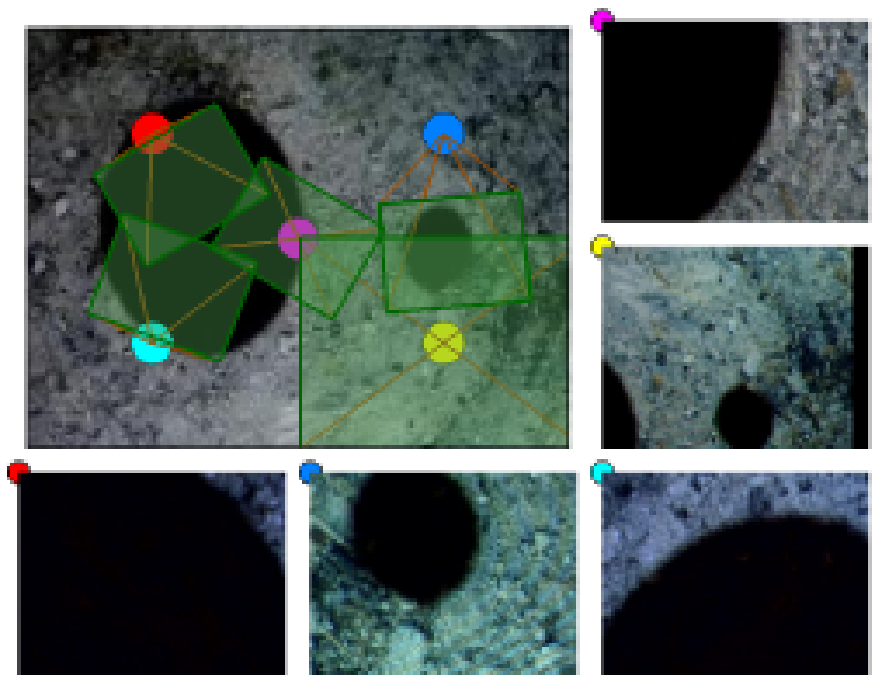}}
\caption{Snapshots of coverage area and acquired images for Fig. \ref{fig:22c}}
\label{fig:23b}
\end{center}
\end{figure}

The experimental results are shown in 
Figs. \ref{fig:24}, \ref{fig:23a} and  \ref{fig:23b}.
Fig. \ref{fig:24} illustrates
the evolution of the global objective function.
We can confirm from this figure that the actions are basically selected so as to maximize the
global objective function.
Figs. \ref{fig:23a} and \ref{fig:23b} show
the snapshots of the coverage area and the acquired images 
at the times marked on Fig. \ref{fig:24}, where green boxes
on the top left (large) pictures
describe the field of views.

In Fig. \ref{fig:23a}(b), sensor $v_5$ (red) widely covers the environment
by choosing zoom-out mode ($\lambda_2 = 6.8$mm) and $v_4$ (blue) captures the half of the hole,
with zoom-in mode,
which drive up the global objective function.
We see from Fig. \ref{fig:23a}(c) that 
$v_2$ (yellow) covers the remaining half of
the hole and also covers the unmonitored area,
which also increases the objective function.
Then, after a while, they reach a desirable configuration
in Fig. \ref{fig:23a}(d), where the hole is monitored by a sensor in zoom-in mode
and the remaining sensors achieve wide-ranging coverage
by choosing zoom-out mode while avoiding overlaps of field of views.

Fig. \ref{fig:23b}(a) illustrates the configuration 
at the time when the image in Fig. \ref{fig:22c}
starts to be projected.
We see from Fig. \ref{fig:23b}(b) that $v_5$ (red) 
monitors the hole in zoom-in mode,
and from Fig. \ref{fig:23b}(c) that $v_3$ (cyan)
also takes a similar action.
Then, after the pan and tilt angles are finely tuned
to avoid overlaps, they eventually reach the desirable configuration
depicted in Fig. \ref{fig:23b}(d).

All the above results show the effectiveness of the present approach.
Let us finally emphasize that the ideal results are achieved without using
any prior information of environmental changes 
on when, where and how the changes occur.

\section{Conclusions}
\label{sec:7}

In this paper, we have investigated a cooperative environmental monitoring
for PTZ visual sensor networks and presented a distributed solution to
the problem
based on game theoretic cooperative control and payoff-based learning.
We first have presented a novel optimal
environmental monitoring problem. 
Then, after constituting a potential game via an existing
utility design technique,
we have presented a payoff-based learning algorithm
based on \cite{MYAS_SIAM09} so that the vision sensors are led to not
just a Nash equilibrium but the potential function maximizes.
Finally, we have run experiments
to demonstrate the effectiveness of the present approach.

The authors would like to thank Mr. S. Mori for his
contributions in the experiments. 

\appendices
\section{Proof of Lemma \ref{lem:A1}}    
\label{app:1}

Condition {\bf (A2)} in Definition \ref{def:4}
is straightforward from the structure of PHPIP.
We thus prove only  {\bf (A1)} and {\bf (A3)} below.

Consider a feasible transition $z^1\rightarrow z^2$ with 
$z^1=(a^0,a^1)\in \BB$ and $z^2=(a^1,a^2)\in \BB$, and
partition the set of sensors $\V$ 
according to their behaviors along with the transition as 
\begin{eqnarray}
\Lambda _1\!\!&\!\!=\!\!&\!\!\{v_i\in \V|\ U_i(a^1)\geq U_i(a^0),\ a^2_i\in \C_i(a_i^1)\setminus \{a_i^1\}\},\
\Lambda _2=\{v_i\in \V|\ U_i(a^1)\geq U_i(a^0),\ a^2_i=a_i^1\},
\nonumber\\
\Lambda _3\!\!&\!\!=\!\!&\!\!\{v_i\in \V|\ U_i(a^1)< U_i(a^0),\ a^2_i\in \C_i(a_i^1)\setminus \{a_i^0, a_i^1\}\},
\nonumber\\
\Lambda _4\!\!&\!\!=\!\!&\!\!\{v_i\in \V|\ U_i(a^1)< U_i(a^0),\ a^2_i= a_i^1\},\
\Lambda _5=\{v_i\in \V|\ U_i(a^1)< U_i(a^0),\ a^2_i= a_i^0\}.
\nonumber
\end{eqnarray}
Then, the probability of transition $z^1\rightarrow z^2$
is described by
\begin{eqnarray}
P_{z^1z^2}^\varepsilon = \prod_{i\in \Lambda _1} 
\frac{\varepsilon}{| \C_i(a_i^1)| -1} \prod_{i\in \Lambda _2} 
(1-\varepsilon)\prod_{i\in \Lambda _3} \frac {\varepsilon}
{| \C_i(a_i^1)| - {\delta_i}}
 \prod_{i\in \Lambda _4}(1-\varepsilon)\kappa\varepsilon^{\Delta _i}
\prod_{i\in \Lambda_5}(1-\varepsilon)(1-\kappa\varepsilon^{\Delta _i})
\label{eqn:A.1}
\end{eqnarray}
where $\delta_i=1$ if $a_i^0=a_i^1$ and $\delta_i=2$ otherwise.
We see from (\ref{eqn:A.1}) that
the resistance $\chi(z^1\rightarrow z^2)$ of transition $z^1\rightarrow z^2$ defined in (\ref{eq:2.3}) is equal to 
$ |\Lambda _1| +|\Lambda _3| +\sum_{i\in \Lambda _4}\Delta _i$ since   
\begin{eqnarray}
0<\lim_{\varepsilon\rightarrow 0} 
\frac {P_{z^1z^2}^\varepsilon}
{\varepsilon^{|\Lambda_1| +|\Lambda _3| +\sum_{i\in \Lambda_4}\Delta _i}}
=\prod _{i\in \Lambda_1}\frac {1}{|\C_i(a_i^1)|-1}\prod _{i\in \Lambda_3}
\frac {1}{|\C_i(a_i^1)| - h_i} \kappa^{|\Lambda_4|}<\infty
\label{eqn:A.2}
\end{eqnarray}
holds.
Thus, {\bf (A3)} in Definition \ref{def:4} is satisfied.

Let us next check {\bf (A1)} in Definition \ref{def:4}. 
From the rule of taking exploratory actions in 
Algorithm 1 and the second item of Assumption \ref{ass:1},
we immediately see that the set of the states accessible 
from any $z\in {\mathcal B}$ is equal to ${\mathcal B}$. 
This implies that the perturbed Markov process 
$\{P^\varepsilon\}$ is irreducible.
We next check aperiodicity of $\{P^\varepsilon\}$. 
It is clear that 
any state in ${\rm diag}({\mathcal A})$ has period 1. 
Let us next pick any 
$(a^0,a^1)$ from the set ${\mathcal B}\setminus {\rm diag}({\mathcal A})$.
Since $a_i^0\in \C_i(a_i^1)$ holds iff $a_i^1\in \C_i(a_i^0)$
from Assumption \ref{ass:1},
the following two paths are both feasible:
$(a^0,a^1)\rightarrow (a^1,a^0)\rightarrow (a^0,a^1)$, 
$(a^0,a^1)\rightarrow (a^1,a^1)\rightarrow (a^1,a^0)\rightarrow (a^0,a^1)$.
This implies that the period of state $(a^0,a^1)$ is $1$ and
the process $\{P^\varepsilon\}$ is proved to be aperiodic.
Hence the process $\{P^\varepsilon\}$ is both 
irreducible and aperiodic, which means {\bf (A1)} in Definition \ref{def:4}.

\section{Proof of Lemma \ref{lem:A2}}    
\label{app:2}

Because of the rule at Step 2 of PHPIP, it is clear that 
any state belonging to $\diag( {\mathcal A})$ cannot move to another
state without explorations, which implies that 
all the states in $\diag(\A)$ itself form recurrent communication 
classes of the unperturbed Markov process $\{P^0 \}$.

Let us consider the states in ${\mathcal B}\setminus \diag({\mathcal A})$
and prove that such states are never included in the recurrent communication 
classes of the unperturbed process $\{P^0\}$.
Here, we use induction.
We first consider $n=1$.
If $U_1(a_1^1)\geq U_1(a_1^0)$, then 
the transition $(a_1^0,a_1^1)\rightarrow (a_1^1,a_1^1)$ is taken.
Otherwise, a sequence of transitions 
$(a_1^0,a_1^1)\rightarrow (a_1^1,a_1^0) \rightarrow (a_1^0,a_1^0)$ occurs.
Thus, for $n = 1$, the state 
$(a_1^0,a_1^1)\in {\mathcal  B}\setminus\diag({\mathcal A})$ is never included 
in recurrent communication classes of $\{P^0 \}$.

We next make a hypothesis that there exists a $n' \in {\mathbb Z}_+$ such 
that all the states in ${\mathcal B}\setminus \diag({\mathcal A})$ are not 
included in recurrent communication classes of 
the unperturbed Markov process $\{P^0\}$ for all $n\leq n'$.
Then, we consider the case $n = n'+1$, where
there are three possible cases:
\begin{description}
\item[(i)] $U_i(a^1)\geq U_i(a^0)\ {\forall i}\in \V = \{1,\cdots, n'+1\}$,
\item[(ii)] $U_i(a^1)< U_i(a^0)\ {\forall i}\in \V = \{1,\cdots, n'+1\}$,
\item[(iii)] $U_i(a^1)\geq U_i(a^0)$ for $n''$ agents where
$n'' \in \{2,\cdots, n'\}$.
\end{description}
In case (i), the transition 
$(a^0,a^1)\rightarrow (a^1,a^1)$ must occur for $\varepsilon= 0$ and,
in case (ii), the transition
$(a^0,a^1)\rightarrow (a^1,a^0)\rightarrow (a^0,a^0)$ should be selected.
Thus, all the states in ${\mathcal B}\setminus \diag({\mathcal A})$
satisfying (i) or (ii) are 
never included in recurrent communication classes.

In case (iii), at the next iteration, 
all the agents $i$ satisfying $U_i(a^1)\geq U_i(a^0)$
choose the current action.
Then, such agents possess a single action
in the memory and, 
in case of $\varepsilon= 0$, each agent has to choose either of
the actions in the memory.
Namely, these agents never change their actions in all subsequent iterations.
The resulting situation is thus the same as the case of $n = n' + 1 - n''$.
From the above hypothesis, we can conclude that
the states in case (iii) are also not included in recurrent communication classes. 
In summary, the states in
${\mathcal B}\setminus \diag({\mathcal A})$ are never included 
in the recurrent communication classes of $\{P^0 \}$.
The proof is thus completed.

\section{Proof of Lemma \ref{lem:A3}}    
\label{app:3}

Along with the straight route, the sensor $v_i$ such that 
$a^1_i \neq a^2_i$ first 
explores from $a_i^1$ to 
$a_i^2$, whose probability is 
%
$(1-\varepsilon)^{n-1}{\varepsilon}/{(|\C_i(a_i^0)| - 1)}$.
%
This implies that the resistance of the transition 
$z^1 = (a^1,a^1)\rightarrow (a^1,a^2)$ is equal to $1$.

We next consider the transition from $(a^1,a^2)$ to $z^2 = (a^2,a^2)$.
If $U_i(a^2)\geq U_i(a^1)$ is true, the probability of this transition 
is $(1-\varepsilon)^{n}$,
whose resistance is equal to $0$.
Otherwise, the inequality $U_i(a^2)<U_i(a^1)$ holds and the 
probability of this transition is equal to 
$(1-\varepsilon)^{n}\times \kappa\varepsilon^{\Delta_i}$,
whose resistance is $\Delta_i$.
See Fig. \ref{fig:14} for the graphic description
of the above sentences.
Let us now notice that the resistance $\chi(\rho)$ 
of the straight route $\rho$ is equal to the sum of the resistances 
of transitions $(a^1,a^1)\rightarrow (a^1,a^2)$ and $(a^1,a^2)\rightarrow 
 (a^2,a^2)$ from (\ref{eqn:resistance_path1}),
and that $\Delta_i < 1/2$ from Assumption \ref{ass:3}.
Hence, we can conclude that $\chi(\rho)$ is smaller than $3/2$.

Let us next prove that the above resistance is minimal among
all paths from $z^1$ to $z^2$.
Suppose now that there is a path $\rho'$ other than the straight route $\rho$
such that $\chi(\rho') < \chi(\rho) < 3/2$. 
Then, the path can accept only one 
exploration of one sensor since two explorations lead to resistance $2$.
We see from Algorithm \ref{alg:1} that any sensor 
with $a_i(k-1) = a_i(k-2)$ would not take 
an action other than $a_i(k-1)$ without exploration 
regardless of the other sensors' actions.
Thus, the sensor taking exploration has to be $v_i$
such that $a_i^1 \neq a_i^2$.

If we denote the chosen action through the exploration
by $a_i'$, then the available joint action in the future is limited to 
$a^1$ and $(a_i', a^1_{-i})$ since no exploration
will be taken.
Thus, in order that $z^2$ will be chosen in the future,
$a_i'$ must be equal to $a_i^2$, and then either of $a^1$ and $a^2$
can occur afterward.
Accordingly, the only way to reach $z^2$ at a round is to follow
the transition $(a^1,a^2)\rightarrow (a^2,a^2)$,
whose resistance is the same as $\chi(r) - 1$.
This contradicts the assumption of $\chi(\rho') < \chi(\rho)$,
and hence the proof is completed.

\section{Proof of Lemma \ref{lem:A4}}    
\label{app:4}

As shown in Appendix \ref{app:3}, the resistance of a straight route
must be equal to $1$ or $1 + \Delta_i \in (1, 3/2)$.
Suppose now that the route $\rho$ contains $M_{\rho}$ straight routes with
resistance greater than $1$, and $\rho'$ contains $M_{\rho'}$ such straight routes.
Let us also denote by $v_{i_\iota}$ the sensor taking exploration along with the straight route  
$z^{(\iota)}\Rightarrow z^{(\iota+1)}$ in $\rho$. 
Then, the sensor $v_{i_\iota}$ also takes exploration
along with $z^{(\iota)}\Leftarrow z^{(\iota+1)}$ in $\rho'$. 
We also use the notations 
\begin{equation*}
\tilde{\Delta}_{{\iota}} := U_{i_{\iota}}(a^{(\iota)}) - U_{i_{\iota}}(a^{(\iota+1)}),\
\tilde{\Delta}'_{{\iota}} := -\tilde{\Delta}_{{\iota}}.
\end{equation*}
Since $z^{(\iota)}\Rightarrow z^{(\iota+1)}$ is a straight route
and hence only $v_{i_{\iota}}$ changes his action along with the route,
the following equation holds from Lemma \ref{lem:1} and (\ref{eqn:potential_game}).
\begin{eqnarray}
\tilde{\Delta}_{{\iota}}= U_{i_{\iota}}(a^{(\iota)}) - U_{i_{\iota}}(a^{(\iota+1)})
= \phi(a^{(\iota)}) - \phi(a^{(\iota+1)}) 
\label{eqn:B.1}
\end{eqnarray}
From the proof of Lemma \ref{lem:A3},
the resistance of $z^{(\iota)}\Rightarrow z^{(\iota+1)}$ in $\rho$
should satisfy
\begin{eqnarray}
\chi(z^{(\iota)}\Rightarrow z^{(\iota+1)}) 
= \left\{
\begin{array}{ll}
1,&\mbox{if }U_{i_{\iota}}(a^{(\iota+1)})\geq U_{i_{\iota}}(a^{(\iota)})\\
1 + \tilde{\Delta}_{{\iota}} \in (1, 3/2), & \mbox{if }
U_{i_{\iota}}(a^{(\iota+1)})<U_{i_{\iota}}(a^{(\iota)}),
\end{array}
\right.
\nonumber
\end{eqnarray}
while the resistance of $z^{(\iota)}\Leftarrow z^{(\iota+1)}$ in $\rho'$
is given as
\begin{eqnarray}
 \chi(z^{(\iota)}\Leftarrow z^{(\iota+1)}) 
= \left\{
\begin{array}{ll}
1+\tilde{\Delta}_{\iota}' \in 
 (1, 3/2),&\mbox{if }U_{i_{\iota}}(a^{(\iota+1)})\geq U_{i_{\iota}}(a^{(\iota)})\\
1, & \mbox{if }
U_{i_{\iota}}(a^{(\iota+1)})<U_{i_{\iota}}(a^{(\iota)}).
\end{array}
\right.\nonumber
\end{eqnarray}
Namely, either of the resistances of 
$z^{(\iota)} \Rightarrow z^{(\iota+1)}$ and $z^{(\iota+1)} \Leftarrow z^{(\iota)}$
is exactly $1$ and the other is greater than $1$ (Fig. \ref{fig:27})
except for the case that $U_i(a^{(\iota+1)}) = U_i(a^{(\iota)})$
in which the resistances are both equal to $1$.
Let us now collect all the $\tilde{\Delta}_{{\iota}}$ such that
the resistance
of $z^{(\iota)}\Rightarrow z^{(\iota+1)}$ is greater than $1$
and number them as
$\bar{\Delta}_1, \cdots, \bar{\Delta}_{M_{\rho}}$.
Similarly, we define $\bar{\Delta}'_1, \cdots, \bar{\Delta}'_{M_{\rho'}}$
for the reverse route $\rho'$.
Then, from (\ref{eqn:B.1}), we obtain
\begin{equation}
\Delta_1 +  \cdots + \Delta_{M_{\rho}} - 
(\Delta'_1 +  \cdots + \Delta'_{M_{\rho'}}) 
= \phi(a^1)-\phi(a^2).
\label{eqn:B.2}
\end{equation}
Note that (\ref{eqn:B.2}) holds
even in the presence of pairs $(a^{(\iota)},a^{(\iota+1)})$ such that
$U_{i_{\iota}}(a^{(\iota+1)}) = U_{i_{\iota}}(a^{(\iota)})$.
Since $\Delta_1+\cdots+\Delta_{M_{\rho}}=\chi(\rho)-(M-1)$ 
and $\Delta'_1+\cdots+\Delta'_{M_{\rho'}}=\chi(\rho')-(M-1)$ 
from (\ref{eqn:resistance_path1}), we obtain
%
$\chi(\rho) = \chi(\rho') + \phi(a^0) - \phi(a^1)$,
%
which means the statement of this lemma.

\begin{figure}
\begin{center}
\includegraphics[width=12cm]{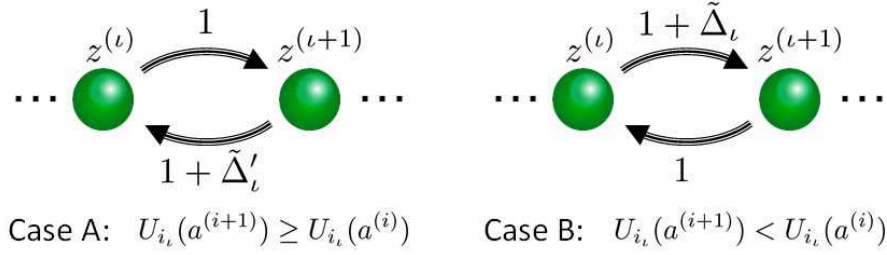}
\caption{Resistance of a straight route (The numbers around
arrows describe resistances of paths)}
\label{fig:27}
\end{center}
\end{figure}

\section{Proof of Lemma \ref{lem:A5}}    
\label{app:5}

The edges of $G_R$, denoted by $\E_R$, are divided into 
$\E_{s}$ in (\ref{eqn:set_es}) and $\E_d := \E_R\setminus \E_s$.
From Lemma \ref{lem:A3}, the weights of the edges in $\E_s$ 
are smaller than $3/2$.
We next consider the weights 
of an edge from $z^1 = (a^1,a^1)\in \diag (\A)$
to $z^2 = (a^2,a^2)\in \diag(\A)$ such that $(z^1,z^2)\in \E_d$.
Then, there exist more than $2$ sensors 
such that $a^1_i \neq a^2_i$, or 
only one sensor $v_i$ such that $a^1_i \neq a^2_i$
satisfies $a^2_i \notin \C_i(a^1_i)$.
In both cases, at least two explorations must happen to
reach $z^2$ and hence
the resistance of any path in $\E_d$
has to be greater than $2$.
Namely, we have
\begin{equation}
w_{l_sl_s'} < w_{l_d, l_d'}\ \ {\forall} (H_{l_s}, H_{l_s'})\in \E_s 
\mbox{ and } (H_{l_d}, H_{l_d'})\in \E_d.
\label{eqn:EsEd}
\end{equation}

We next form a graph 
$G_R' = (\H, \E_R, W_R')$ by just reversing
the weights of all edges over graph $G_R$.
Namely, the weight $w_{ll'}$ on $G_R$
is equal to  the weight $w_{l'l}$ on $G_R'$.
Let us now apply the Chu-Liu/Edmonds Algorithm \cite{E_BJRNBS67}
to the graph $G_R'$ and compute the minimal tree with a root $H_R$
such that there is a unique path from $H_R$ to any node 
(the directions of edges are opposite to the tree defined in Subsection \ref{sec:5.1}).
Then, it is not difficult to confirm that reversing the directions of all 
edges of the minimal tree yields the minimal resistance tree with root $H_R$ over $G_R$.
Hence, it is sufficient to prove that 
the Chu-Liu/Edmonds Algorithm provides a tree 
consisting only of edges in $\E_{s}$.

In the algorithm \cite{E_BJRNBS67}, every node in $\H \setminus \{H_R\}$ initially chooses the
incoming edge with the minimal weight.
Then, only edges in $\E_s$
can be chosen at the initial step from (\ref{eqn:EsEd}).
If the resulting graph $\T_0$ consisting only of such edges is acyclic,
then the minimum spanning tree is formed and the statement of this lemma is true.
Otherwise, there is at least one cycle in $\T_0$.

We next focus on one of such cycles denoted by $G_{\rm cyc} = (\H_{\rm cyc}, 
 \E_{\rm cyc}, \W_{\rm cyc})$, where all edges in $\E_{\rm cyc}$ have to
be contained in $\E_s$.
In the following, the weight of the edge in $\E_{\rm cyc}$ 
entering a node $H_l \in \H_{\rm cyc}$ is denoted by $w_{\rm cyc}^l$,
and we define $\bar{w}_{\rm cyc} := \min_{H_l \in \H_{\rm cyc}}w_{\rm cyc}^l$.
Then, each node $H_l$ in $\H_{\rm cyc}$ computes 
the temporal weights for all edges from $H_{l'}\notin \H_{\rm cyc}$
to $H_l$ over $G_R'$ by
\begin{equation}
 \tilde{w}_{ll'} = w'_{ll'} - w_{\rm cyc}^l + \bar{w}_{\rm cyc},
\label{eqn:temporal_weight}
\end{equation}
and identifies a node $H_{l'}$ providing the minimal $\tilde{w}_{ll'}$, where
such a node $H_{l'}$ is denoted by $\bar{H}_l$
and the corresponding $\tilde{w}_{ll'}$ is denoted by $\bar{w}_{l}$.
Then, we seek $H_{l^*}\in {\mathcal H}_{\rm cyc}$ 
with the minimal $\bar{w}_{l}$
and replace the edge entering $H_{l^*}$ over $\T_0$  by 
the edge $(H_{l^*}, H_l)$.

Chu-Liu/Edmonds Algorithm repeats the above process and eventually
finds the minimum spanning tree.
Namely, if we can prove that $(H_{l^*}, H_l)$ must be included in $\E_s$,
the statement of the lemma is true.
Now, notice that there exists at least one $H_{l} \in \H_{\rm cyc}$
such that the set
\[
 \H_l = \{H_{l'} \in \H \setminus \H_{\rm cyc}|\
(H_{l'},H_l) \in \E_s\}
\]
is not empty from Assumption \ref{ass:1}.
For such a node $H_l$, the node $\bar{H}_l$ must be chosen from $\H_l$
since (\ref{eqn:EsEd}) holds and
the second and third terms in (\ref{eqn:temporal_weight}) are common
for all options of $H_{l'} \notin \H_{\rm cyc}$.
Then, $w'_{ll'}$ in (\ref{eqn:temporal_weight}) must be smaller than 
$3/2$ since $(H_{l'},H_l) \in \E_s$, and 
$- w_{\rm cyc}^l + \bar{w}_{\rm cyc} \in (-1/2,0]$ holds
since $w_{\rm cyc}^l \in [1, 3/2)$, $\bar{w}_{\rm cyc} \in [1, 3/2)$,
and $\bar{w}_{\rm cyc} := \min_{H_l \in \H_{\rm cyc}}w_{\rm cyc}^l$.
Namely, $\bar{w}_l$ must be smaller than $3/2$
for all $H_l$ such that $\H_l$ is not empty.
In contrast, for any $H_l$ such that $\H_l$ is empty,
$w'_{ll'}$ in (\ref{eqn:temporal_weight}) must be greater than or equal to $2$
because of $(H_l, H_{l'}) \in \E_d$.
Then, by using $- w_{\rm cyc}^l + \bar{w}_{\rm cyc} \in (-1/2,0]$ again,
$\bar{w}_l$ is proved to be greater than $3/2$
and such $H_{l'}$ is never selected as $H_{l^*}$.
Namely, $(H_{l^*}, H_l)$ must be contained in $\E_s$.
This completes the proof.

\end{document}